\documentclass[citeautoscript,aps,prl,twocolumn,superscriptaddress]{revtex4-1}  
\usepackage{graphicx}  
\usepackage{dcolumn}   
\usepackage{bm}        
\usepackage{amssymb}   
\usepackage{amsmath}
\usepackage{amsthm}
\usepackage{amsfonts}
\usepackage[caption=false]{subfig}
\usepackage{dsfont}
\newtheorem{theorem}{Theorem}
\usepackage{braket}
\usepackage{hyperref}
 \usepackage{physics}
\usepackage{xr}
\usepackage{mathtools}
\DeclarePairedDelimiter\ceil{\lceil}{\rceil}
\DeclarePairedDelimiter\floor{\lfloor}{\rfloor}
\usepackage{multibib}
\newcites{A}{References for Supplemental Material}

\usepackage[usenames]{xcolor}

\begin{document}
\title{Bounding the finite-size error of quantum many-body dynamics simulations}
\author{Zhiyuan Wang}
\affiliation{Department of Physics and Astronomy, Rice University, Houston, Texas 77005,
  USA}
\affiliation{Rice Center for Quantum Materials, Rice University, Houston, Texas 77005, USA}
\author{Michael Foss-Feig}
\affiliation{Honeywell $|$ Quantum Solutions}
\author{Kaden R.~A. Hazzard}
\affiliation{Department of Physics and Astronomy, Rice University, Houston, Texas 77005,
USA}
\affiliation{Rice Center for Quantum Materials, Rice University, Houston, Texas 77005, USA}
\date{\today}

\begin{abstract}
  Finite-size error~(FSE), the discrepancy between an observable in a finite system and in the thermodynamic limit, is ubiquitous in numerical simulations of quantum many body systems. Although a rough estimate of these errors can be obtained from  a sequence of finite-size results, a strict, quantitative bound on the magnitude of FSE is still missing. Here we derive rigorous upper bounds on the FSE of local observables in real time quantum dynamics simulations initialized from a product state. In $d$-dimensional locally interacting systems with a finite local Hilbert space, our bound implies $ |\langle \hat{S}(t)\rangle_L-\langle \hat{S}(t)\rangle_\infty|\leq C(2v t/L)^{cL-\mu}$, with $v$, $C$, $c$, $\mu $ constants independent of $L$ and $t$, which we compute explicitly. For periodic boundary conditions~(PBC), the constant $c$ is twice as large as that for open boundary conditions~(OBC), suggesting that PBC have smaller FSE than OBC at early times. The bound can be generalized to a large class of correlated initial states as well. As a byproduct, we prove that the FSE of local observables in   ground state simulations decays exponentially with $L$, under a suitable spectral gap condition. Our bounds are practically useful in determining the validity of finite-size results, as we demonstrate in simulations of the one-dimensional~(1D) quantum Ising and Fermi-Hubbard models.
  
\end{abstract}

\maketitle
\paragraph{Introduction}
Numerical simulations are crucial to our  understanding of many-body quantum matter, and are routinely applied in all fields of physics and in chemistry. %
Unfortunately, many  numerical techniques popular in these fields incur significant FSEs when approximating properties of a large~(potentially infinite) system by properties of a finite one.
The most direct example is exact diagonalization (ED), which exactly solves the finite system  numerically~\cite{laflorencie:simulations_2004,noack:diagonalization_2005,
sandvik:computational_2010,laeuchli:introduction_2011}.  	
Accessible system sizes are limited since the Hilbert space dimension grows exponentially with system size; for the simplest case of interacting spin-1/2s, a state-of-the-art ground state calculation is limited to  $\sim45$ spins~\footnote{Reaching even these system sizes is possible only if internal, translation, and point group symmetries are utilized, and if state-of-the-art algorithms and large-scale computational resources are employed. Researchers usually employ much smaller systems for computational convenience.}.
FSEs also significantly affect  other techniques, such as density matrix renormalization group (DMRG)~\cite{white:density_1992,hallberg2006new,schollwoeck:density_2011,stoudenmire2012studying},
many tensor network algorithms~\cite{Perez2007matrix,orus:practical_2014}, quantum dynamical typicality-based algorithms~\cite{bartsch2009dynamical,elsayed2013regression,steinigeweg2014pushing,steinigeweg2014spin,steinigeweg2014scaling}, and quantum Monte Carlo (QMC)~\cite{nightingale:quantum_1999}, and they are a significant source of  error for simulating quantum systems on quantum computers~\cite{childs:toward_2018}   and  for analog quantum simulations using ultracold matter~\cite{bloch:quantum_2012}, trapped ions~\cite{blatt:quantum_2012}, and other platforms~\cite{altman2019quantum}. 

It is often difficult to characterize  FSEs. The standard method to assess them is to calculate and compare observables for different system sizes, ideally using finite-size scaling~\cite{sandvik2010computational}. Although useful, this method has limitations. One is that it offers no guarantees. Two different system sizes may have results that closely agree, but at larger sizes the physics changes and observables  deviate~\cite{bausch2018size}. Another is that one may not be able to study multiple system sizes that are sufficiently large to get a good estimate of the convergence. 

 In this paper, we 
 derive \textit{rigorous} upper bounds on the error of approximating observables in a large, possibly infinite, quantum many-body system by results in a smaller one.  The bounds are applicable to arbitrary Hamiltonians for which a Lieb-Robinson~(LR)  bound exists. For quantum dynamics simulations starting from product initial states and evolving under locally interacting Hamiltonians with a finite local Hilbert space, the bound for a local observable ${\hat S}$ is  
 \begin{equation}\label{eq:dynamicEB}
 |\langle \hat{S}(t)\rangle_L-\langle \hat{S}(t)\rangle_\infty|\leq C (2v t/L)^{c L-\mu},
 \end{equation}
 where $v$, $C$, $c$, and $\mu$ are constants that can be computed explicitly and  depend on the Hamiltonian, observable, and boundary condition. 
 Such dynamics is explored in a wide variety of ultracold matter experiments, such as quantum quenches and slow ramps in  Rydberg atoms~\cite{zeiher:many_2016,takei:direct_2016,bernien2017probing,guardado-sanchez:probing_2018,lienhard:observing_2018,orioli:relaxation_2018}, molecules~\cite{yan:observation_2013,hazzard:many_2014,seeselberg:extending_2018},  Fermi gases~\cite{smale:observation_2019}, atoms in optical lattices~\cite{de2013nonequilibrium, meldgin:probing_2016,choi:exploring_2016, bordia2017probing, gabardos2020relaxation}, and optical clocks~\cite{goban:emergence_2018}. This dynamics can probe fundamental phenomena, such as many-body localization~\cite{eisert2015quantum, nandkishore:many_2015,luitz2016extended,luitz2017ergodic,parameswaran:many_2018}, prethermalization~\cite{mori:thermalization_2018,schmied:non-thermal_2018}, and generation of topological defects near critical points~\cite{simon:quantum_2011}. 
 This bound is then extended to a large family of correlated initial states satisfying an exponential clustering condition. While our main focus is on  dynamics, we also 
 show that the FSE of local observables in a many-body ground state decays exponentially in system size, under a suitable spectral gap condition.
 
 The idea behind our bound is that locality -- specifically that one piece of a system does not instantly affect far-away pieces -- imposes strong constraints on quantum dynamics~\cite{bravyi2006,hastings2010locality}. This can be seen by considering evolution under a Hamiltonian initiated from a product state~(other scenarios can be understood by similar arguments). As illustrated in Fig.~\ref{fig:basicconfig}, an observable in a region $X$ will be affected by FSEs only after a long enough time for information to propagate from the  boundary to $X$. This idea is made precise by relating FSE to unequal time correlation functions, which can then be bounded by a LR bound~\cite{Lieb1972}, a direct consequence of locality. Although similar ideas of applying LR bounds to analyze the performance of some numerical algorithms have been employed in Refs.~\cite{Osborne2006,Osborne2007a,Osborne2007b,kliesch2014lieb,woods2015simulating,woods2016dynamical,haah2018quantum,tran2019locality,huang2020computing}, the connection to FSE has not been made explicit, and the practical utility of the bounds for numerics was not demonstrated. This idea has also been applied to estimate FSE in a non-rigorous way, for example in Ref.~\cite{daug2020dynamical}.
 
\begin{figure}
\center{\includegraphics[width=0.9\linewidth]{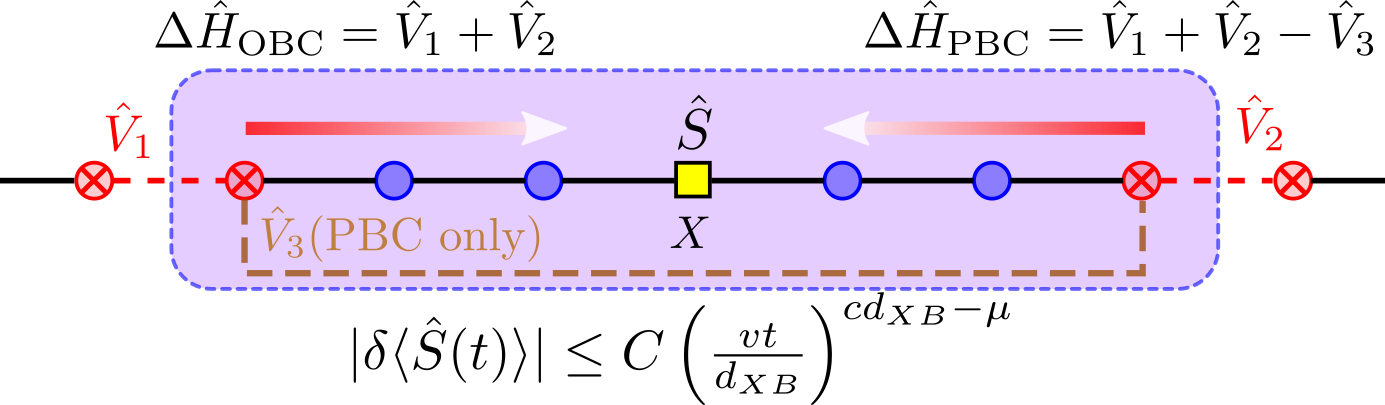}}
\caption{\label{fig:basicconfig} An illustration of our bounds. 
	In locally interacting systems, information propagates no faster than the LR speed $v$, so it takes a finite amount of time $t_c\sim d_{XB}/v$ for the effect of the boundary, $\Delta\hat{H}$ to affect the center site observable $\hat{S}$. Here $d_{XB}$ is the distance between the supports $X$~(yellow square) and $B$~(red crossed circles) of $\hat S(t)$ and $\Delta {\hat H}$, respectively. 
	}
\end{figure}  

Our FSE bound not only shows the convergence of finite-size approximations in principle, but is tight enough to be useful in practice, which we demonstrate in simulations of some prototypical models. For example, in the dynamics following a sudden change of parameters in a 1D transverse field Ising model~(TFIM) with $L=21$ sites, 
the error bounds for the transverse magnetization and nearest-neighbor correlations remain extremely small to times where they have evolved close to equilibrium. Furthermore, the bounds are reasonably tight: the time at which the error bound becomes significant is only 20--25\% smaller than the time at which the actual FSE becomes noticeable. We similarly demonstrate this 
for the non-equilibrium relaxation of the Fermi-Hubbard model (FHM) from a checkerboard
state, inspired by experiments and theory of Refs.~\cite{trotzky:probing_2012,bauer:temporal_2015}. The precision  of these bounds is enabled by the major quantitative improvements offered by recent LR bounds~\cite{Lucas2019operator,wang:tightening_2019}.

In addition to their quantitative utility, these bounds provide insights into the convergence of numerical methods, and  open the way  to designing new algorithms. One immediate consequence of the bounds is to rigorously show that the FSE decays exponentially with the linear dimension of the system
 for periodic boundary condition~(PBC), as well as for open boundary condition~(OBC) provided that one measures observables only near the center of the system, as commonly employed in the DMRG community.
If one instead averages the measurement over all sites in OBC, then our bound indicates that the error decays only algebraically. Similar behavior at finite temperature has been observed and analyzed in Ref.~\cite{iyer:optimization_2015}. Furthermore, if one compares PBC to OBC with center site measurement, our error bound for PBC decays twice as fast with distance as the bound for OBC at early times, suggesting that PBC gives more reliable results at early times~\footnote{While the fact that the error bound for PBC is smaller than that for OBC does not necessarily imply that the actual FSE in PBC is smaller, in the Supplemental Material~\cite{Suppl} we use short-time perturbative arguments to show that the actual error in PBC is indeed smaller at early times for most initial product states.}.
These insights may lead to new methods; one example is that they show why the moving-average cluster expansion (MACE) method of Ref.~\cite{hazzard:many_2014} converges exponentially faster than alternative schemes.

\paragraph{A simple bound for both OBC and PBC}

Consider the dynamical evolution of a quantum many-body system on an infinite $d$-dimensional lattice, governed by a locally interacting Hamiltonian $\hat{H}$. For illustrative purpose, in Fig.~\ref{fig:basicconfig} we draw the configuration for a 1D nearest-neighbor interacting lattice model.
Let $|\psi\rangle$ be the initial product state, $\hat{S}$ be a local observable to be measured that acts on a finite region $X$~(center point in Fig.~\ref{fig:basicconfig}), and let $\Delta \hat{H}=\sum_j \hat{V}_j$ be the sum of all the interaction terms between the inner and outer parts of the system~(red links in Fig.~\ref{fig:basicconfig}). If PBC are used, we further subtract from $\Delta\hat{H}$ the interaction between the first and the last site~(brown link in Fig.~\ref{fig:basicconfig}). Let $\hat{H}_L$ and $|\psi_L\rangle$ denote the Hamiltonian and the initial state of the finite-size simulation, respectively~(i.e. the restriction of $\hat{H}$ and $|\psi\rangle$ to the $L$-site inner system). Denote $\hat{H}'=\hat{H}-\Delta\hat{H}$, so that $\hat{H}'$ decouples into two commuting terms, one acting only on the inner system, the other acting only on the outer system. The FSE of the observable $\hat{S}$ is
\begin{equation}\label{eq:deltadef}
\delta \langle \hat{S}(t)\rangle_{\psi} \equiv|\langle e^{i\hat{H}_L t} \hat{S} e^{-i\hat{H}_L t}\rangle_{\psi_L}-\langle e^{i\hat{H}t}\hat{S} e^{-i\hat{H}t} \rangle_{\psi}|.
\end{equation}
where $\langle A\rangle_{\psi}\equiv \langle\psi|A|\psi\rangle$, and we set $\hbar=1$ throughout. Since $\hat{H}'$ decouples into two independent spatial regions (inner and outer) and $|\psi\rangle$ is a product state, the first term in Eq.~\eqref{eq:deltadef} can  be rewritten as $\langle e^{i\hat{H}' t} \hat{S} e^{-i\hat{H}' t}\rangle_{\psi}$. Inserted into Eq.~\eqref{eq:deltadef}, the two expectation values are taken in the same state $|\psi\rangle$, so their difference can be bounded by the operator norm $|\langle \psi|\hat{A}|\psi\rangle|\leq \|\hat{A}\|$. 
Using the unitary invariance of operator norm $\|\hat{A}\|=\|\hat{U}\hat{A}\hat{V}\|$ for arbitrary unitary operators $\hat{U},\hat{V}$, we have
\begin{equation} \label{eq:MLR2}
|\delta \langle \hat{S}(t)\rangle_{\psi}|\leq \| \hat{U}_I(t)\hat{S}\hat{U}_I(t)^\dagger-\hat{S} \|.
\end{equation}
where $\hat{U}_I(t)=e^{-i\hat{H}t}e^{i\hat{H}'t}$ is the evolution operator in the interaction picture, which satisfies $\hat{U}_I(0)=1$ and
$i\partial_t \hat{U}_I(t)=\hat{U}_I(t)\Delta\hat{H}(t)$,
where $\Delta\hat{H}(t)=e^{-i\hat{H}'t}\Delta\hat{H} e^{i\hat{H}'t}$.  Now applying the fundamental theorem of calculus and the triangle inequality, we obtain a bound on the FSE:
\begin{eqnarray}\label{eq:deltaS0tintro}
  |\delta \langle \hat{S}(t)\rangle_{\psi}|&\leq& \int^t_0\| \frac{d}{dt'}[\hat{U}_I(t')\hat{S}\hat{U}_I(t')^\dagger-\hat{S}] \|dt'\nonumber\\
  &=& \int^t_0\| \hat{U}_I(t')[\Delta\hat{H}(t'),\hat{S}]\hat{U}_I(t')^\dagger \|dt'\nonumber\\
  &=& \int^t_0\| [\Delta\hat{H}(t'),\hat{S}] \|dt'.
\end{eqnarray}
The integrand is the quantity bounded by LR bounds, so to upper bound the FSE, one can insert the relevant LR bound. We focus on locally interacting systems, but Eq.~\eqref{eq:deltaS0tintro} applies equally to long-range interactions by substituting the corresponding LR bounds~\cite{hastings2006,Gorshkov2014,Gong2014persistence,Foss2015nearly,tran2019locality,Lucas2019longrange,kuwahara2020strictly,tran2020hierarchy} in those systems. For a locally-interacting system, the currently tightest LR bound is obtained by computing the series in Eq.~\eqref{eq:Lucas_bound_commu} of the Supplemental Material~(SM)~\cite{Suppl}, which is based on Refs.~\cite{Lucas2019operator,wang:tightening_2019}, although this may not be efficiently computable in general. A slightly looser but efficiently computable method is discussed in Ref.~\cite{wang:tightening_2019}, in which one solves a system of first order linear differential equations for a number of variables proportional to the system size. To see the qualitative features of the bound for large systems, we can insert the simple expression given in Eq.~(3) of Ref.~\cite{wang:tightening_2019} into Eq.~\eqref{eq:deltaS0tintro} to obtain
\begin{equation}\label{eq:factorialLR}
  |\delta \langle \hat{S}(t)\rangle_{\psi}|\leq  \sum_{j} c_j\left(\frac{v|t|}{d_{Xj}}\right)^{D(\hat{S},\hat{V}_j)},
\end{equation}
where $c_j$ are constants independent of $t$ and $d_{Xj}$, $D(\hat{S},\hat{V}_j)$ is the distance between the operators $\hat{S},\hat{V}_j$ in the commutativity graph~(CG) as introduced in Ref.~\cite{wang:tightening_2019}, and $v$ is the LR speed. The distance in the CG is related to the distance in real space $d_{Xj}$ by $D(\hat{S},\hat{V}_j)= \eta d_{Xj}-\mu$, where $\eta,\mu$ are~(straightforwardly determined) constants, and $d_{Xj}$ is the distance between $X$ and $j$ in real space. Therefore the rhs of Eq.~\eqref{eq:factorialLR} is bounded by $(\frac{vt}{d_{XB}})^{\eta d_{XB}-\mu}$, where $d_{XB}=\min_{j\in B} d_{Xj}$. Despite its simplicity, the  $t$-dependence of this bound generically agrees with   the exact error to lowest order in $t$ in OBC~\cite{Suppl}.

Besides its practical utility for bounding FSE in calculations, as demonstrated below, this result has qualitative implications. One is to rigorously support the common practice of measuring observables close to the center site in OBC numerics~(e.g. in the DMRG community), rather than averaging over all sites. This minimizes the error bound, since the center size maximizes $d_{XB}$. This choice yields our main result in Eq.~\eqref{eq:dynamicEB} for the OBC case. Our bound allows one to extend this. For example, in dimension greater than one, we can minimize FSE by choosing an optimal cluster shape that minimizes the rhs of Eq.~\eqref{eq:factorialLR}, and run simulations on the optimal shape.

\paragraph{An improved bound for PBC}
In the previous section we treated PBC in a way similar to OBC. But it turns out that the resulting bound in Eqs.~\eqref{eq:deltaS0tintro} and \eqref{eq:factorialLR} is qualitatively loose at small $t$ for PBC. The reason for this can be intuitively understood as follows. The two terms in the rhs of Eq.~\eqref{eq:deltadef} can be expanded in $t$. As we discuss in greater detail in the SM~\cite{Suppl}, 
the FSE for $\hat{S}(t)$ actually is only contributed by terms in $e^{i\hat{H}t}\hat{S} e^{-i\hat{H}t}$ whose spatial span is larger than $L$ and terms in $e^{i\hat{H}_L t} \hat{S} e^{-i\hat{H}_L t}$ that wrap around the whole periodic system. 
The leading order of these terms is proportional to $t^{\mathcal L}$, where ${\mathcal L}$ is the length of the shortest non-contractible loop on the PBC commutativity graph, which is roughly twice as large as the exponent $D(\hat{S},\hat{V}_j)$ in Eq.~\eqref{eq:factorialLR}. The SM~\cite{Suppl} extends methods developed in Refs.~\cite{Lucas2019operator,wang:tightening_2019} to derive a rigorous upper bound for $|\delta\langle\hat{S}(t)\rangle|$ that leads to this improved $t^{\mathcal{L}}$ scaling. The main result is
\begin{equation}\label{eq:deltaSPBC}
|\delta \langle \hat{S}(t)\rangle^{(\text{PBC})}_{\psi}|\leq  \sum_{1\leq p\leq d}C_p \left(\frac{2v_p t}{ L_p}\right)^{\mathcal{L}_p},
\end{equation}
where the constant $C_p$ is given in Eq.~\eqref{eq:const_Cp}, $v_p$ is the LR speed in the $p$-th direction given in Eq.~\eqref{eq:directionalv_p}, and  ${\mathcal L}_p$ is the size of the periodic system in the $p$-th direction in commutativity graph. $\mathcal{L}_p$ is related to the real space system size $L_p$ by $\mathcal{L}_p= \eta_p L_p-\mu_p$ for constant integers $\eta_p,\mu_p$. 
We note that while this bound improves the small-time exponent of the PBC bound by a factor of $2$ compared to Eqs.~\eqref{eq:deltaS0tintro} and \eqref{eq:factorialLR}, the timescale $t_c\approx \min_p L_p/2v_p$ on which the bound exponentially grows is still approximately the same as Eq.~\eqref{eq:factorialLR}. Besides its quantitative utility, Eq.~\eqref{eq:deltaSPBC} shows that in anisotropic systems where $v_p$ is different in each direction, one should choose $L_p\propto v_p$ in order to minimize the FSE.  %

\paragraph{FSE in non-degenerate gapped ground states}
So far, we have been discussing FSEs of quantum dynamics simulations. We now derive a bound on FSE of local observables in non-degenerate ground states under a gap assumption. This result is interesting in its own right, and will also be useful for our subsequent generalization of the dynamics error bound to correlated initial states. 

The Hamiltonian $\hat{H}$, observable $\hat{S}$, boundary terms $\Delta\hat{H}$, etc. are the same as before. For convenience we suppose that the operator  $\hat{S}=\hat{S}_l$ has unit norm and acts nontrivially only within a region $X=X_l$ of diameter $l$ which sits on the center of the finite-size cluster. The difference now is that we consider the observable $\langle\hat{S}\rangle=\mathrm{Tr}[\hat{\rho} \hat{S}]$ given by the ground state density matrix $\hat{\rho}$. The numerical simulation approximates this thermodynamic quantity by the expectation value in the finite size ground state $\mathrm{Tr}[\hat{\rho}_L \hat{S}]$, where $\hat{\rho}_L$ is the ground state density matrix of $\hat{H}_L$. Our result relies on an assumption that the interpolated Hamiltonain $\hat{H}(\lambda)\equiv \hat{H}-\lambda\Delta\hat{H}$ is non-degenerate for all $0\leq \lambda\leq 1$ and has a uniform spectral gap $\min_{0\leq \lambda\leq 1}\Delta(\lambda)=\Delta>0$. When this condition is satisfied, then analogously to Eq.~\eqref{eq:dynamicEB} we have~\cite{Suppl}
\begin{equation}\label{eq:initialstateassumption}
	|\mathrm{Tr}[\hat{\rho}\hat{S}_l-\hat{\rho}_L\hat{S}_l]|\leq C e^{-(L-l)/2\xi}.
\end{equation}

\paragraph{Bounds for correlated initial states}
We now generalize our error bound to dynamics initiated from  a class of (possibly mixed) initial states $\hat{\rho}$, for which there exists a sufficiently good finite size approximation $\hat{\rho}_L$ satisfying Eq.~\eqref{eq:initialstateassumption}. This includes non-degenerate gapped ground states~(under the condition described above), but also includes translation invariant matrix product states~(MPS) with a finite bond dimension~\footnote{For translation invariant MPS with a finite bond dimension, $\hat{\rho}_L$ can be taken as the $L$-site periodic version of $\hat{\rho}$. That $\hat{\rho}_L$ satisfies the condition in Eq.~\eqref{eq:initialstateassumption} can be proved using the transfer operator method which is used to prove that MPS has finite correlation length, see, e.g. Refs.~\cite{schollwoeck:density_2011,Perez2007matrix,orus:practical_2014}. The parameter $2\xi$ can simply be taken as the correlation length of the MPS.}, and finite temperature thermal states~\cite{kliesch2014locality} $\hat{\rho}=e^{-\beta \hat{H}}/\mathrm{Tr} [e^{-\beta \hat{H}}]$ above a certain temperature, where
$\hat{\rho}_L = e^{-\beta \hat{H}_L}/\mathrm{Tr} [e^{-\beta \hat{H}_L}]$, i.e. the thermal state of $ \hat{H}_L$. %

Given that we have an initial state $\hat{\rho}_L$ satisfying Eq.~\eqref{eq:initialstateassumption}, we can bound the dynamics FSE as
\begin{eqnarray}\label{eq:deltaS_correlatedrho}
	\delta \langle \hat{S}(t)\rangle_{\rho} &=&|\langle  \hat{S}_L(t) \rangle_{\rho_L}-\langle \hat{S}(t) \rangle_{\rho}|\\
	&\leq& |\mathrm{Tr}[(\hat{\rho}_L-\hat{\rho}_{[L]})\hat{S}_L(t)]|+|\langle  \hat{S}_L(t)- \hat{S}(t) \rangle_{\rho}|\nonumber,
\end{eqnarray}
where $\hat{S}_L(t)=e^{i\hat{H}_L t} \hat{S} e^{-i\hat{H}_L t}$, $\hat{\rho}_{[L]}$ is the reduced density matrix of $\hat{\rho}$ on the finite cluster, and in the second line we used the triangle inequality.
The second term can be bounded using the same method as in Eq.~\eqref{eq:deltadef}, since $\mathrm{Tr}[\hat{\rho}\hat{A}]\leq \|\hat{A}\|$ for any density matrix $\hat{\rho}$. To bound the first term, we insert the expansion $\hat{S}_L(t)=\hat{S}_0(t)+\sum^L_{l=1}[\hat{S}_l(t)-\hat{S}_{l-1}(t)]$ into Eq.~\eqref{eq:deltaS_correlatedrho},
and notice that  $\hat{S}_l(t)-\hat{S}_{l-1}(t)$ is an operator acting on $X_l$, whose norm is bounded by Eqs.~(\ref{eq:deltaS0tintro},\ref{eq:factorialLR}) to be $\|\hat{S}_l(t)-\hat{S}_{l-1}(t)\|\leq C(2vt/l)^{\eta l/2}$ for some constant $C$. For initial states satisfying Eq.~\eqref{eq:initialstateassumption}, this implies
\begin{eqnarray}\label{eq:deltaS_correlatedrho_res}
	\delta \langle \hat{S}(t)\rangle_{\rho} \leq C_1 e^{(vt-L/2)/\xi}+C_2 e^{\eta(vt-L/2)},
\end{eqnarray}
where $C_1$ and $C_2$ are model-dependent constants that can be explicitly determined.%

\paragraph{Example: 1D TFIM}
We test our dynamics error bounds in simulations of prototypical models for quantum many-body physics, starting with
the TFIM,
\begin{equation}\label{eq:HIsing}
\hat{H}=-J\sum_{j}\hat{\sigma}^z_j\hat{\sigma}^z_{j+1}-h\sum_{j}\hat{\sigma}^x_j.
\end{equation}
This is a canonical model for quantum phase transitions~\cite{vojta2003quantum,sachdev2007quantum}, and occurs in materials like CoNb$_2$O$_6$~\cite{coldea2010quantum}, 
 cold atom~\cite{simon2011quantum,labuhn2016tunable,guardado2018probing} and trapped ion~\cite{friedenauer2008simulating,kim2010quantum,kim2011quantum,lanyon2011universal,britton2012engineered} experiments, and superconducting circuits~\cite{barends2016digitized,harris2018phase}. 
We numerically study the  dynamics of this model at the critical point $J=h$ for several $L$,  and calculate the exact evolution for $L=\infty$. Specifically, we study the dynamics of $\langle \hat{\sigma}^x(t)\rangle$ starting from  $|\psi(0)\rangle=|\rightarrow\rightarrow\ldots\rightarrow\rangle$. Analogous  dynamics in the 2D TFIM has been explored in Rydberg atom experiments~\cite{guardado-sanchez:probing_2018,lienhard:observing_2018}.
 
Fig.~\ref{fig:mainfig}a shows $\langle\hat{\sigma}^x_j(t)\rangle$    from $L=5$ to $21$ using PBC, along with the exact $L=\infty$ solution~\cite{dziarmaga2005dynamics}. 
To obtain a FSE bound for $\langle\hat{\sigma}^x_j(t)\rangle$, we use the LR bound given in 
Eq.~\eqref{eq:TFIMcommutator_bound} of Ref.~\cite{Suppl}~[obtained from the general bound Eq.~\eqref{eq:Lucas_bound_dbcommu}],
which, after inserting into Eq.~\eqref{eq:deltaS0tintro}, yields
\begin{equation}\label{eq:errIsing}
 |\delta \langle\hat{\sigma}^x_j(t)\rangle|\leq 4\sqrt{\frac{J}{h}}\frac{(2\sqrt{Jh}t)^{2L-1}}{(2L-1)!}+\frac{J}{h}\frac{(4\sqrt{Jh}t)^{2L-2}}{(2L-2)!}.
\end{equation}
As Fig.~\ref{fig:mainfig}a shows, this error bound provides a guarantee of the numerical calculations' accuracy out to interesting and useful timescales. For the $L=21$-site calculation, the bound guarantees that the results are accurate~(within $10^{-2}$) up to times $Jt\sim 3.5$, where the observable has nearly reached equilibrium. Furthermore, this time 
is in reasonable accord with the true time at which FSE becomes important~(within 20\%).

We emphasize that the FSE  bound never made use of the TFIM's exact solution. The bound Eq.~\eqref{eq:deltaS0tintro} can be applied to any system including in dimensions greater than one. 
As we will now demonstrate in the 1D FHM, the bound still provides a useful guarantee of the finite-size results when no exact solution is available.

\begin{figure}
	\includegraphics[width=\columnwidth]{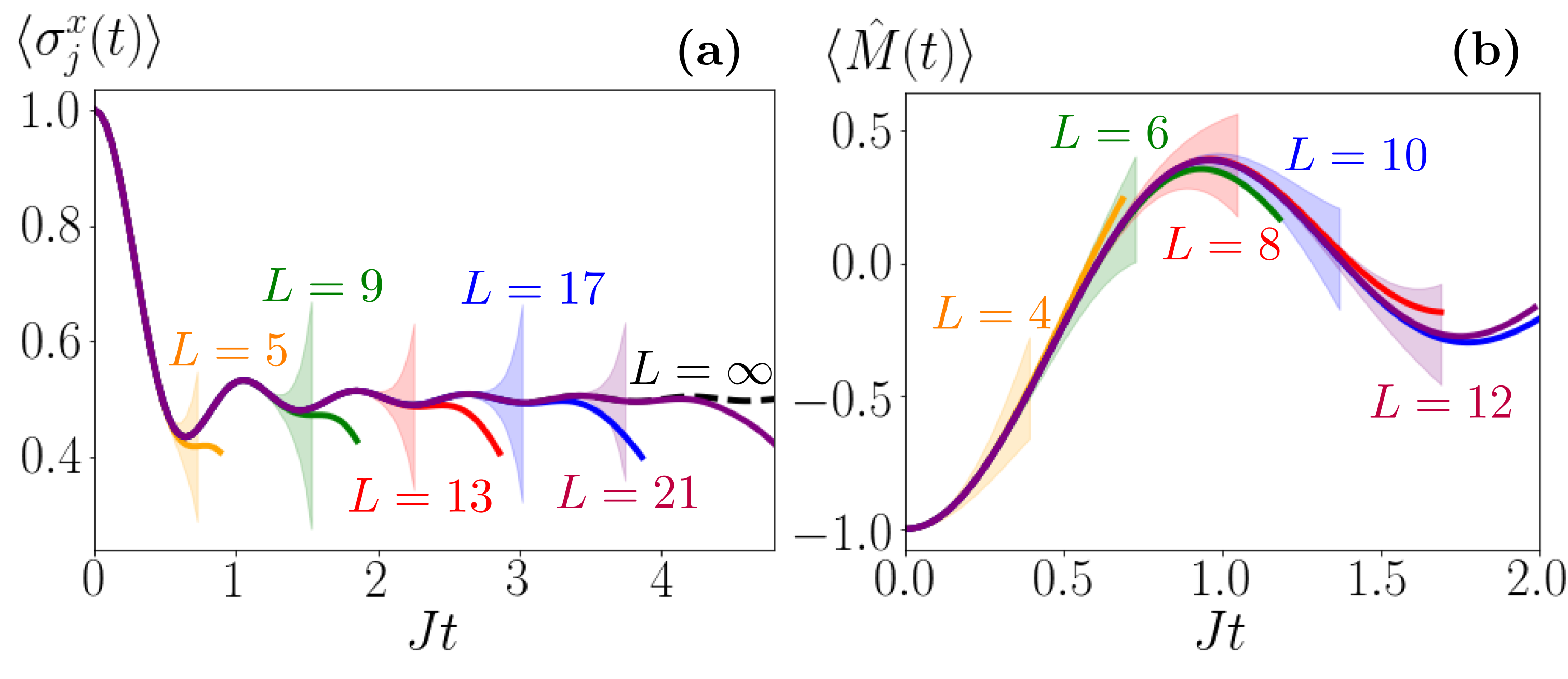}
	\caption{\label{fig:mainfig}Numerically exact evolution for (a) $\langle\hat{\sigma}^x_j(t)\rangle$ in the $L$-site PBC TFIM at $J=h$, with $|\psi(0)\rangle=|\rightarrow\rightarrow\ldots\rightarrow\rangle$, and (b) $\langle\hat{M}(t)\rangle$ in the $L$-site OBC FHM at $U=0.5 J$, with $|\psi(0)\rangle=|202020\ldots20\rangle$. The dashed curve in TFIM is the exact $L\to\infty$ result. The shaded areas for each curve represent the region in which the actual $L\to\infty$ values must lie according to the FSE bounds in Eq.~\eqref{eq:errIsing} and Eqs.~(\ref{eq:deltaS0tintro}, \ref{eq:Lucas_bound_commu}), respectively.}
\end{figure}

\paragraph{Example: FHM}
The 1D FHM describes  spin-1/2 fermions in a lattice whose Hamiltonian is
\begin{equation}\label{eq:FHH}  
  \hat{H}=-J\sum_{\langle ij\rangle,\sigma=\uparrow,\downarrow}( \hat{a}^\dagger_{i\sigma}\hat{a}^{\phantom{\dagger}}_{j\sigma}+\mathrm{H.c.})+U\sum_i \hat{n}^\uparrow_{i} \hat{n}^\downarrow_{i}.
\end{equation}
The FHM exhibits rich behavior, such as a metal-Mott insulator transition, and potentially  high-temperature superconductivity. It is a reasonable approximation of some real materials, such as FeO, NiO, CoO~\cite{anisimov1991band}, and has been realized in ultracold atoms~\cite{esslinger2010fermi,parsons2016site,boll2016spin,mazurenko2017cold,brown2019bad}. 

We numerically study the relaxation dynamics of a charge density wave state $|\psi(0)\rangle=|202020\ldots20\rangle$, analogous to previous theory~\cite{FHnonequilibrium} and experiments~\cite{pertot2014relaxation}, where $2$ (0) means a doubly occupied (empty) site.  We run the finite-size simulations in OBC, 
and measure the density imbalance  
$\hat{M}(t)=[\hat{N}_{\mathrm{even}}(t)-\hat{N}_{\mathrm{odd}}(t)]/L$~\cite{FHnonequilibrium}. 
To get a FSE bound for $\hat{M}(t)$, we use the currently tightest LR bound, obtained by numerically summing the series in Eq.~\eqref{eq:Lucas_bound_commu} of the SM~\cite{Suppl}
and inserting the result into Eq.~\eqref{eq:deltaS0tintro}.
Fig.~\ref{fig:mainfig}b  shows the results.

Our error bound can be compared to estimates of FSE obtained from comparing calculations of different sizes. For example, one can take the difference between the $L=10$ and $L=12$ as a rough estimate of the FSE of the $L=12$ calculation.  Our bound is comparable in its guaranteed timescale of convergence to this conventional estimate. For example, we can guarantee that the FSE in $\langle \hat{M}(t)\rangle$ of the $L=12$ result is less than $1\%$ for $Jt\leq 1.2$, comparable to the time $Jt\sim 1.6$ where the $L=10$ and $12$ results differ noticeably. 

\paragraph{Conclusions}
We have presented a rigorous upper bound on the FSE of local observables measured in numerical simulations of quantum dynamics starting from a large class of initial states. For product initial states, the bounds show an advantage of using PBC at early times. We also presented a generalization to simulations of local observables in non-degenerate gapped ground states. 
In all the cases we considered, the bounds decay exponentially in system size, and guarantee the accuracy of finite size dynamics simulation up to a time scale $t_c\sim L/2v$. These insights into FSE can motivate better algorithms.

The quantitative utility of the bounds is demonstrated in the 1D TFIM and 1D FHM. In both cases, the error bounds are extremely small up to timescales where there is interesting physics and even equilibration, and they are reasonably tight compared to the actual FSEs. We expect these bounds to provide useful tools to researchers going forward, providing FSE bounds on numerical calculations and suggesting new numerical methods that minimize this error. 

\acknowledgements
We thank Miles Stoudenmire, Miroslav Hopjan, Bhuvanesh Sundar and Ian White for discussions, and Brian Neyenhuis for a careful reading of the manuscript. This work was supported in part by the Welch Foundation~(C-1872), the National Science Foundation~(PHY-1848304), and the Office of Naval Research~(N00014-20-1-2695).

\bibliography{finite-size-error}
\bibliographystyle{apsrev4-1}

\clearpage
\newpage
\clearpage 
\setcounter{equation}{0}%
\setcounter{figure}{0}%
\setcounter{table}{0}%
\renewcommand{\thetable}{S\arabic{table}}
\renewcommand{\theequation}{S\arabic{equation}}
\renewcommand{\thefigure}{S\arabic{figure}}
\renewcommand{\thesection}{S\arabic{section}}

\newtheorem{proposition}{Proposition}
\onecolumngrid

\begin{center}
	{\Large Supplemental Material for\\ ``Bounding the FSE of quantum many-body dynamics simulations''}
	
	\vspace*{0.5cm}
	
	Zhiyuan Wang, Michael Foss-Feig, and Kaden R. A. Hazzard
	
	\vspace{1cm}
	
\end{center}

The supplemental material fills in technical details of the basic results in the main text.  We first prove the ground state FSE bound in Eq.~\eqref{eq:initialstateassumption} under the gap assumption we mentioned there. Then we give the full derivation of the dynamics FSE bound for correlated initial state, Eq.~\eqref{eq:deltaS_correlatedrho_res}. We then compare our dynamics error bounds to perturbation theory at early times and show that they are qualitatively tight for a large class of product initial states, i.e. they grow with time as $t^a$ with the correct exponent $a$. The remaining task is to give a detailed derivation of the improved error bound in PBC given in Eq.~\eqref{eq:deltaSPBC}. An important intermediate step is to prove Eq.~\eqref{eq:EB_PBC}, which is an analog of Eq.~\eqref{eq:deltaS0tintro}, expressing the FSE in terms of a LR commutator. We give three different methods to numerically upper bound the LR commutator that appears in the rhs of Eq.~\eqref{eq:EB_PBC}. The first one is to evaluate the series in Eq.~\eqref{eq:Lucas_bound_dbcommu}, which combines the methods in Refs.~\cite{Lucas2019operator,wang:tightening_2019}. This leads to the tightest bound, but is only efficiently computable in special cases. The second one is to numerically solve the differential equation \eqref{eq:DEGreen} and then calculate Eq.~\eqref{eq:dbcommufinal}. This method~(which is based on Ref.~\cite{wang:tightening_2019}) is only slightly looser than the first one but is computationally efficient in general. The third method makes further simplifications which result in the simple analytic expression in Eq.~\eqref{eq:deltaSPBC}, whose constants are given in Eqs.~(\ref{eq:const_Cp},\ref{eq:directionalv_p},\ref{eq:deltaS_simple_final}). We emphasize that for any desired system, one may compute the LR bound and use it in the error formulas Eqs.~\eqref{eq:deltaS0tintro} and \eqref{eq:EB_PBC}. In this way, as LR bounds are refined in the future or extended to more general systems (e.g. long-range interacting~\cite{Lucas2019longrange,kuwahara2020strictly,tran2020hierarchy}, bosonic~\cite{schuch2011information}, or continuum ones), these refinements can immediately be used in the FSE bounds.

\subsection{FSEs in gapped non-degenerate ground states}
In this section we prove the exponential decay of FSE in gapped non-degenerate ground states under the assumption stated in the main text. 
Recall that $\hat{H}$ is the Hamiltonian in the thermodynamic limit, and $\hat{H}'=\hat{H}-\Delta\hat{H}$ is obtained from $\hat{H}$ by removing the boundary links. The interpolated Hamiltonian $\hat{H}(\lambda)=(1-\lambda)\hat{H}+\lambda\hat{H}'=\hat{H}-\lambda \Delta\hat{H}$ is assumed to have a non-degenerate ground state $|G(\lambda)\rangle$ and is uniformly gapped $\Delta(\lambda)\geq \Delta>0$, for all $\lambda\in [0,1]$. The basic idea is that,  as the parameter $\lambda$ is varied from 0 to 1, we keep track of how fast $|G(\lambda)\rangle$ changes, as well as $\langle \hat{S}_l\rangle_\lambda\equiv \langle G(\lambda)| \hat{S}_l|G(\lambda)\rangle$. Suppose $|G(\lambda)\rangle$ is properly normalized such that $\langle G(\lambda)|G(\lambda)\rangle=1$ and phase chosen such that $\langle G(\lambda)|\frac{d}{d\lambda}|G(\lambda)\rangle=0$ for any $\lambda\in [0,1]$. Then first order non-degenerate perturbation theory gives
\begin{equation}
	\frac{d}{d\lambda}|G(\lambda)\rangle=-\frac{\bar{P}_G(\lambda)}{E_0(\lambda)-\hat{H}(\lambda)}\Delta\hat{H}|G(\lambda)\rangle,
\end{equation}
where $E_0(\lambda)$ is the ground state energy of $\hat{H}(\lambda)$ and $\bar{P}_G(\lambda)\equiv \hat{\mathds{1}}-|G(\lambda)\rangle\langle G(\lambda)|$ is the projection operator to the space of excited states. Therefore, we have the integral formula for the FSE $\delta S_l$:
\begin{equation}\label{eq:deltaSX}
	\delta S_l\equiv \langle \hat{S}_l\rangle_1-\langle \hat{S}_l\rangle_0=-\int^1_0d\lambda[\langle G(\lambda)|\hat{S}_l\frac{\bar{P}_G(\lambda)}{E_0(\lambda)-\hat{H}(\lambda)}\Delta\hat{H}|G(\lambda)\rangle+\mathrm{H.c.}].
\end{equation}
The integrand  looks similar to  the ground state correlator between $\hat{S}_l$ and $\Delta\hat{H}$, so it's natural to guess that it should decay exponentially in a gapped system. This is  verified by the following theorem~(which is similar to the exponential clustering theorem in non-degenerate gapped ground states~\cite{hastings2006}):
\begin{theorem}
	Let $\hat{A}_X, \hat{B}_Y$ be arbitrary local observables with unit norm, supported on non-overlapping regions $X,Y$, respectively. Let $|G\rangle$ be the unique ground state of locally-interacting Hamiltonian $\hat{H}$ with spectral gap $\Delta$. Then the quantity $f_{XY}\equiv\langle G|\hat{A}_X\frac{\bar{P}_G}{E_0-\hat{H}}\hat{B}_Y|G\rangle+\mathrm{c.c.}$ is upper bounded by
	\begin{equation}\label{eq:WXY}
		|f_{XY}|\leq C\sqrt{d_{XY}}e^{-d_{XY}/\xi}, \text{   ~~~where~~~   } 
			\xi^{-1}= \begin{cases}
				\frac{\Delta}{2v}&\text{ if }\Delta\leq v\\
				\frac{1}{2}W(\frac{\Delta^2e}{v^2})&\text{ if }\Delta> v,
			\end{cases}
	\end{equation}
	where $v$ is the LR velocity, $W(x)$ is the Lambert-$W$ function, $C$ is a constant depending on $\hat{A}_X$ and $\hat{B}_Y$, and $d_{XY}$ is the distance between $X$ and $Y$. 
\end{theorem}
\begin{proof}
	 Using the identity 
	\begin{eqnarray}\label{eq:GaussFourier}
		&&\int^\infty_{-\infty}e^{-\alpha T^2}dT\int^T_0dt \langle G|[\hat{A}_X(t),\hat{B}_Y]|G\rangle\nonumber\\
		&=&i\int^\infty_{-\infty}dT e^{-\alpha T^2} \{\langle G|\hat{A}_X\frac{1-e^{i(E_0-\hat{H})T}}{E_0-\hat{H}}\hat{B}_Y|G\rangle+\mathrm{H.c.}\}\nonumber\\
		&=&i\sqrt{\frac{\pi}{\alpha}}\{\langle G|\hat{A}_X\frac{\bar{P}_G}{E_0-\hat{H}}[1-e^{-\frac{(E_0-\hat{H})^2}{4\alpha}}]\hat{B}_Y|G\rangle+\mathrm{H.c.}\},
	\end{eqnarray}
	where $\alpha>0$ is a parameter to be specified later, we have
	\begin{eqnarray}
		f_{XY}&=&
		\{\langle G|\hat{A}_X\frac{\bar{P}_G}{E_0-\hat{H}}e^{-\frac{(E_0-\hat{H})^2}{4\alpha}}\hat{B}_Y|G\rangle+\mathrm{H.c.}\}\\
		&&-i\sqrt{\frac{\alpha}{\pi}}\int^\infty_{-\infty}e^{-\alpha T^2}dT\int^T_0dt \langle G|[\hat{A}_X(t),\hat{B}_Y]|G\rangle.\nonumber
	\end{eqnarray}
	Therefore
	\begin{eqnarray}\label{eq:proofofthm1}
		|f_{XY}|&\leq&2\frac{e^{-\frac{\Delta^2}{4\alpha}}}{\Delta}
		+\sqrt{\frac{\alpha}{\pi}}\int^\infty_{-\infty}e^{-\alpha T^2}dT\int^T_0dt \|[\hat{A}_X(t),\hat{B}_Y]\|\nonumber\\
		&\leq&2\frac{e^{-\frac{\Delta^2}{4\alpha}}}{\Delta}
		+2\sqrt{\frac{\alpha}{\pi}}\int^\infty_{0}e^{-\alpha T^2}[C\left(\frac{vT}{r}\right)^{r+1}_{\leq 1}+2(T-r/v)_{\geq 0}]dT\nonumber\\
		&\leq&2\frac{e^{-\frac{\Delta^2}{4\alpha}}}{\Delta}
		+2\sqrt{\frac{\alpha}{\pi}}\int^{r/v}_{0}e^{-\alpha T^2}C\left(\frac{vT}{r}\right)^{r+1}dT
		 +2\sqrt{\frac{\alpha}{\pi}}\int^\infty_{r/v}e^{-\alpha T^2}[C+2(T-r/v)]dT\nonumber\\
		&\leq&\frac{2}{\Delta}e^{-\frac{\Delta^2}{4\alpha}}
		+2C\sqrt{\frac{\alpha}{\pi}}\frac{r}{v}\max_{0\leq T\leq r/v}e^{-\alpha T^2}\left(\frac{vT}{r}\right)^{r+1}+C_3 e^{-\alpha r^2/v^2}\nonumber\\
		&=&\frac{2}{\Delta}e^{-\frac{\lambda\Delta^2 r}{4}}+C_2\max_{0\leq \tau\leq 1/v}e^{-\tau^2r/\lambda}\left(v\tau\right)^{r+1}+C_3 e^{-\frac{r}{\lambda v^2}},
	\end{eqnarray}
	where we use the notation $(x)_{\leq 1}=\min\{x,1\},(x)_{\geq 0}=\max\{x,0\}$, $C$ is a constant coming from the LR bound that does not depend on $r,\alpha$, $C_2,C_3$ are coefficients that only weakly depend on $r,\alpha$, and in the last line we substituted  $T=\tau r$ and $\alpha=1/(\lambda r)$. 
	
	Notice that Eq.~\eqref{eq:proofofthm1} holds for arbitrary positive $\lambda$, since the parameter $\alpha>0$ introduced in Eq.~\eqref{eq:GaussFourier} can be chosen arbitrarily. If we choose $\lambda$ to be any positive value, we can immediately prove the exponential decay of $|f_{XY}|$ in $d_{XY}$ which leads to Eq.~\eqref{eq:initialstateassumption} in the main text, since all the three terms in the rhs of Eq.~\eqref{eq:proofofthm1} decay exponentially in $r$. If we want a better bound, we can choose $\lambda$ to maximize the smallest decay coefficient $\min\{\lambda\Delta^2 /4,\min_{1\leq \tau\leq 1/v}[\tau^2/\lambda-\ln(v\tau)],1/(\lambda v^2)\}$,  to make the rhs of Eq.~\eqref{eq:proofofthm1} decay in $r$ as fast as possible.  For $\Delta\leq v$ we choose $\lambda=2/(\Delta v)$ and the maximum is at $\tau=1/v$, while
	for $\Delta>v$, we choose $\lambda$ to be the solution to the equation $\lambda \Delta^2/2+\ln(\lambda v^2/2)=1,$ and the maximum occurs at $\tau=\sqrt{\lambda/2}$. In the end we arrive at Eq.~\eqref{eq:WXY}. 
	This finishes the proof of our theorem.
\end{proof}

Inserting Eq.~\eqref{eq:WXY} into Eq.~\eqref{eq:deltaSX}, we get
\begin{equation}\label{eq:deltaSXfin}
	|\delta S_l|=|\mathrm{Tr}[\hat{\rho}\hat{S}_l-\hat{\rho}_L\hat{S}_l]|\leq C e^{-(L-l)/2\xi'},
\end{equation}
where $\xi'$ is chosen to be slightly larger than the $\xi$ in Eq.~\eqref{eq:WXY} to compensate the $\sqrt{d_{XY}}$ prefactor, and the constant $C$ is adjusted accordingly. This proves Eq.~\eqref{eq:initialstateassumption} in the main text. 

\subsection{Proof of Eq.~\eqref{eq:deltaS_correlatedrho_res}: bound for dynamics initiated from a correlated initial state}
We assume $2vt<L$, since these are the only times we will apply the FSE bounds; at longer times the error bounds become too large to be useful. We have 
\begin{eqnarray}
	|\mathrm{Tr}[(\hat{\rho}_L-\hat{\rho}_{[L]})\hat{S}_L(t)]|&=&|\mathrm{Tr}\{(\hat{\rho}_L-\hat{\rho}_{[L]})\sum^L_{l=0}[\hat{S}_l(t)-\hat{S}_{l-1}(t)]\}|\nonumber\\
&\leq&\sum^L_{l=0}\|\hat{S}_l(t)-\hat{S}_{l-1}(t)\||\mathrm{Tr}[(\hat{\rho}_L-\hat{\rho}_{[L]})\hat{O}_l]|\nonumber\\
&\leq&  C_1\sum^{\floor{2vt}}_{l=0}  e^{-(L-l)/2\xi} +C_2\sum^L_{l=\ceil{2vt}} (2vt/l)^{\eta l/2} e^{-(L-l)/2\xi} \nonumber\\
&\leq & C'_1  e^{(vt-L/2)/\xi}+C_2\sum^L_{l=\ceil{2vt}} e^{\eta(vt-l/2)} e^{-(L-l)/2\xi} \nonumber\\
&= & C'_1  e^{(vt-L/2)/\xi}+C'_2  e^{\eta(vt-L/2)},
\end{eqnarray}
where $\hat{O}_l$ is a unit norm operator that only acts nontrivially in $X_l$, and $\floor{x},\ceil{x}$ are the floor and ceiling of $x$, respectively. The first term in the third line comes from the trivial bound $\|\hat{S}_l(t)-\hat{S}_{l-1}(t)\|\leq 2\|\hat{S}\|$, while the second term comes from the LR bound $\|\hat{S}_l(t)-\hat{S}_{l-1}(t)\|\leq C(2vt/l)^{\eta l/2}$. In the fourth line we use the inequality $(2vt/l)^{\eta l/2}\leq e^{\eta(vt-l/2)}$ to facilitate the calculation, and we sum the geometric series in the last two lines, with constants $C'_1, C'_2$ depending at most weakly on $t,L$.
Combined with the second term of Eq.~\eqref{eq:deltaS_correlatedrho}, we obtain  Eq.~\eqref{eq:deltaS_correlatedrho_res}. One can explicitly compute as needed the constants appearing in this bound for a given model and initial state $\hat{\rho}$.

\subsection{Comparison of FSE bounds to perturbation theory at early times}
We can gain insights into the accuracy of our error bounds by comparing them to the true FSE at lowest order in time. We can do this analytically using the Taylor series expansion of the true FSE. Recall from  Eq.~\eqref{eq:deltadef} that FSE is defined as
\begin{equation}\label{eq:deltadef_S}
\delta \langle \hat{S}(t)\rangle_{\psi} \equiv|\langle e^{i\hat{H}_L t} \hat{S} e^{-i\hat{H}_L t}\rangle_{\psi_L}-\langle e^{i\hat{H}t}\hat{S} e^{-i\hat{H}t} \rangle_{\psi}|,
\end{equation}
where $\hat{H}_L$ is the $L$-site finite Hamiltonian and $\psi_L$ is the initial product state restricted to this finite chain.  The Taylor series expansion of   the time evolved operators can be expressed as sums of Lie clusters~(nested commutators) involving $\hat{S}$ and terms in the Hamiltonian~[see, e.g. Eq.~\eqref{eq:CBH}]. Many small clusters appear in both $\langle e^{i\hat{H}_L t} \hat{S} e^{-i\hat{H}_L t}\rangle_{\psi_L}$ and $\langle e^{i\hat{H}t}\hat{S} e^{-i\hat{H}t} \rangle_{\psi}$ and therefore cancel each other. Only those clusters whose spatial length is at least half of the system size may contribute to FSE.
In the following we treat OBC and PBC separately, starting with OBC. 

In OBC, FSE is due to those nested commutators in $\langle e^{-i\hat{H}t} \hat{S} e^{-i\hat{H}t}\rangle$ in which boundary terms appear at least once, since only such clusters are not canceled by any cluster in $\langle e^{-i\hat{H}_L t} \hat{S} e^{-i\hat{H}_Lt}\rangle$. The lowest order cluster containing at least a boundary term is proportional to $\langle \hat{C}(\hat{S},\hat{V}_j) \rangle_\psi t^{D(\hat{S},\hat{V}_j)}/D(\hat{S},\hat{V}_j)!$, where $\hat{V}_j$ the boundary term closest to $\hat{S}$, and $\hat{C}(\hat{S},\hat{V}_j)$ is the shortest nested commutator joining $\hat{S}$ and $\hat{V}_j$. Assuming that the expectation value $\langle \hat{C}(\hat{S},\hat{V}_j)\rangle_\psi$ does not vanish~(which is true for most initial product states $|\psi\rangle$), the true FSE has the same $t$ dependence as the bound Eq.~\eqref{eq:factorialLR}.

In PBC there is a qualitative difference. Any term in the Taylor expansion of $\langle e^{i\hat{H}t}\hat{S} e^{-i\hat{H}t} \rangle_{\psi}$ in Eq.~\eqref{eq:deltadef_S} whose spatial span is smaller than $L$ do not contribute to FSE, because they cancel the corresponding terms in  $\langle e^{i\hat{H}_L t} \hat{S} e^{-i\hat{H}_L t}\rangle_{\psi_L}$ due to translation invariance and the product nature of the initial state, as shown in Fig.~\ref{fig:lucasTFIMPBC}. Only those terms in $\langle e^{i\hat{H}t}\hat{S} e^{-i\hat{H}t} \rangle_{\psi}$ which are too long to be embeddable into an $L$-site periodic system can contribute. More precisely, in PBC, the rhs of Eq.~\eqref{eq:deltadef_S} is contributed by the following two classes of terms:\\
(1) terms in $\langle e^{i\hat{H}t}\hat{S} e^{-i\hat{H}t} \rangle_{\psi}$ whose spatial span is larger than $L$, i.e. terms that are too long to be embeddable in the $L$-site PBC chain;\\
(2) terms in $\langle e^{i\hat{H}_L t} \hat{S} e^{-i\hat{H}_L t}\rangle_{\psi_L}$ which wrap around the whole periodic system~(wraps around the torus, or the circle in $d=1$). Such terms do not generally cancel any terms in $\langle e^{i\hat{H}t}\hat{S} e^{-i\hat{H}t} \rangle_{\psi}$.\\
The leading order term in both classes are proportional to $t^{\mathcal L}$, where $\mathcal{L}$ is the number of Hamiltonian terms in the smallest $L$-site-unembeddable cluster starting from $\hat{S}$. This $t^{\mathcal L}$ behavior is the same as the improved PBC bound in Eq.~\eqref{eq:deltaSPBC}, and the exponent $\mathcal{L}$ is roughly twice as large as the OBC exponent $D(\hat{S},\hat{V}_j)$. In the 
next section, we give a rigorous bound for the sum of all terms in each class.
\begin{figure}[h]
	\center{\includegraphics[width=0.5\linewidth]{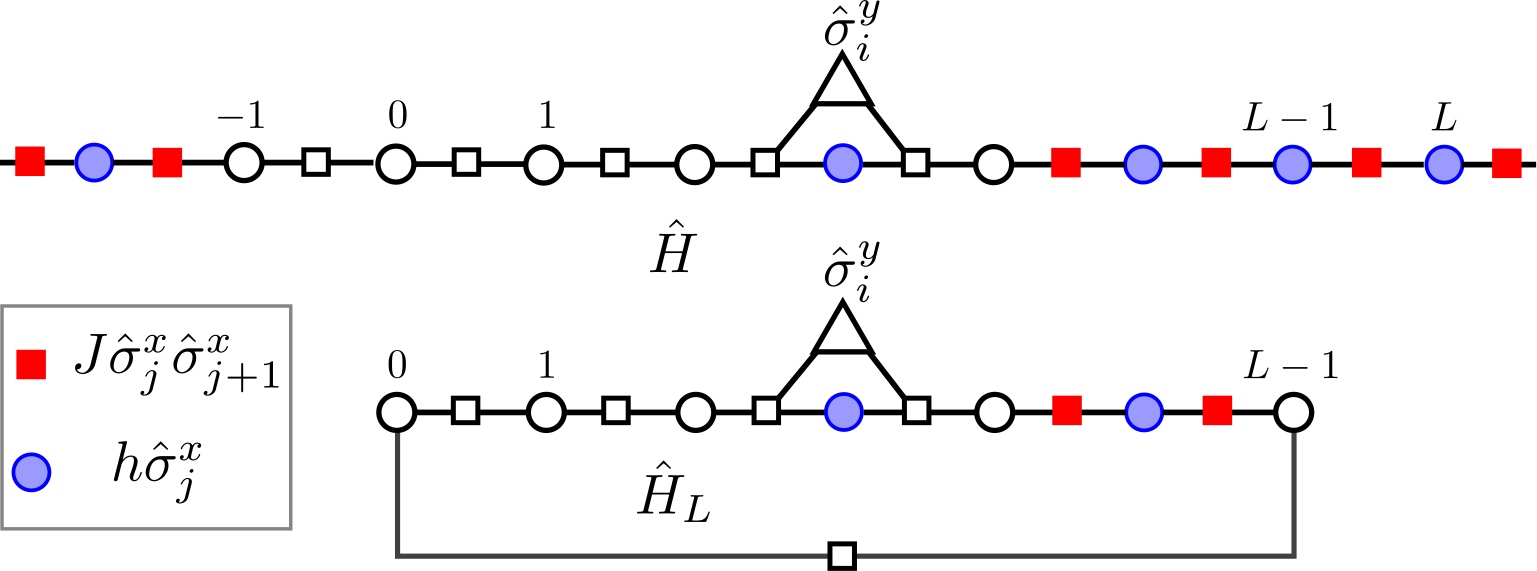}}
	\caption{\label{fig:lucasTFIMPBC} An example of canceling terms in  the Taylor expansions of $e^{i\hat{H}t}\hat{\sigma}^y_i e^{-i\hat{H}t}$ and $e^{i\hat{H}_L t} \hat{\sigma}^y_i e^{-i\hat{H}_L t}$, in the case of 1D TFIM. Vertices represent terms ${\hat h}_j$ of the Hamiltonian $\hat{H}=\sum_j {\hat h}_j$, as well as the observable $\hat{S}=\hat{\sigma}_i^y$, and edges are drawn between terms that do not commute.
		The spatial span of the two Lie clusters shown in this figure~(cluster of open circles, squares, and triangles) are equal and shorter than the system size. The expectation values of these two terms exactly cancel due to translation invariance and the initial state being a product state.}
\end{figure}

\section{Improved bound for PBC}%
	
	For simplicity, we first focus on the 1D case, and later we show that a bound in higher dimension can be obtained by repeatedly using the 1D bound.
	Consider a translation invariant quantum system on a 1D periodic lattice with $L$ unit cells, described by the Hamiltonian $\hat{H}_L$, and let $\hat{H}=\hat{H}_{L\to\infty}$ be the Hamiltonian in the thermodynamic limit. 
	We will derive an upper bound for the sum of all terms in $e^{-i\hat{H}t} \hat{S}e^{-i\hat{H}t}$ that may contribute to the FSE (i.e. whose spatial span is larger than $L$), using ideas motivated by Ref.~\cite{Lucas2019operator}, and then do the same for $e^{-i\hat{H}_Lt}\hat{S}e^{-i\hat{H}_Lt}$.

	Let us focus on the first class of terms, i.e. terms in $ e^{i\hat{H}t}\hat{S} e^{-i\hat{H}t} $ whose spatial span is larger than $L$, since the second class can be treated in an identical way. We write the Hamiltonian in the thermodynamic limit as 
	\begin{equation}\label{eq:H_TD}
	\hat{H}=\sum_{j\in G} \hat{h}_{j},
	\end{equation}
	where $\hat{h}_{j}$ denotes a local term in $\hat{H}$. It is convenient to introduce the notion of the commutativity graph $G$, defined in Ref.~\cite{wang:tightening_2019}. This is a graph whose vertices $j$ are associated with $\hat{h}_j$ and which has edges from $j$ to $j'$ if and only if $\hat{h}_j$ and $\hat{h}_{j'}$ do not commute. The observable $\hat{S}$ is represented as an external vertex $s$ on $G$, linked to all the vertices $j$ whose $\hat{h}_j$ do not commute with $\hat{S}$, as shown in Fig.~\ref{fig:lucasTFIMPBC}. Now we write down the Taylor expansion of $e^{i\hat{H}t}\hat{S} e^{-i\hat{H}t} $. We use bold letters $\mathbf{h}_j=\mathrm{ad}_{\hat{h}_j}$ to denote the adjoint of the corresponding operator $\hat{h}_j$, e.g. $\mathbf{h}_j(\hat{S})\equiv [\hat{h}_j,\hat{S}]$. We have
	\begin{eqnarray}\label{eq:CBH}
	e^{i\hat{H}t}\hat{S} e^{-i\hat{H}t}&=&e^{i\mathbf{H}t}(\hat{S})\nonumber\\
	&=&\sum^\infty_{n=0}\frac{(it)^n}{n!}\sum_{\substack{j_1,\ldots,j_n\in G\\ T(s,j_1,j_2,\ldots,j_n)\in \mathcal{T}_s}}\mathbf{h}_{j_n}\ldots \mathbf{h}_{j_2} \mathbf{h}_{j_1} \hat{S},
	\end{eqnarray}
	where $T(s,j_1,j_2,\ldots,j_n)$ denotes the causal forest of the sequence $(s,j_1,j_2,\ldots,j_n)$, as defined in Ref.~\cite{Lucas2019operator}, and $\mathcal{T}_s$ denotes the set of causal trees starting from the vertex $s$. In simple terms, we are summing over all the non-vanishing connected Lie clusters on $G$ starting from the vertex $s$.
	
	Our goal is to upper bound the sum over all the terms in Eq.~\eqref{eq:CBH} whose spatial span is larger than $L$. For such a Lie cluster, let $j_i$ be the first term in the sequence $j_1,\ldots,j_n$ such that the spatial span of the subsequence $\mathbf{h}_{j_i}\ldots  \mathbf{h}_{j_1} \hat{S}$ is larger than $L$. This means that the spatial span of $\mathbf{h}_{j_{i-1}}\ldots  \mathbf{h}_{j_1} \hat{S}$ is less than or equal to $L$. We call $j_i$ the \textit{earliest unembeddable vertex}~(EUV) of the sequence $(s,j_1,j_2,\ldots,j_n)$. 
	Notice that $\hat{h}_{j_i}$ must act nontrivially on at least two unit cells, because otherwise $j_i$ can never be the EUV of any sequence.
	
	The basic idea is to classify the Lie clusters in Eq.~\eqref{eq:CBH} into different families, with each family having the same EUV, and derive an upper bound for the sum over all terms within each family.
	To this end, let us denote by $[e^{i\mathbf{H}t}(\hat{S})]_k$ the sum of all the $L$-site unembeddable terms in the rhs of Eq.~\eqref{eq:CBH} whose EUV is $k$, i.e.
	\begin{equation}
	[e^{i\mathbf{H}t}(\hat{S})]_k\equiv\sum_{\substack{n\geq 0,\\T(s,j_1,j_2,\ldots,j_n)\in \mathcal{T}_s,\\ \text{EUV}(s,j_1,j_2,\ldots,j_n)=k}}\frac{(it)^n}{n!} \mathbf{h}_{j_n}\ldots \mathbf{h}_{j_2} \mathbf{h}_{j_1} \hat{S}.
	\end{equation}
	
	We limit our explicit proof to the nearest neighbor interacting case in which every term $\hat{h}_j$ in the Hamiltonian acts non-trivially on at most two neighboring unit cells. The proof for the general case is the same as for the nearest-neighbor interacting case, but involves keeping track of more complicated notation. Therefore for the general case we omit the proof and present only the final result. Returning to the nearest-neighbor interacting case, for an arbitrary operator $\hat{A}$, let $x^{\text{L}}_A, x^{\text{R}}_A$ denote the $x$-coordinates of the leftmost and rightmost unit cells on which the operator $\hat{A}$ acts, and let $n_A\equiv x^{\text{R}}_A-x^{\text{L}}_A+1$ denote the spatial span of $\hat{A}$~(we write $x^{\text{L,R}}_k,n_k$ for the operator $\hat{h}_k$ for simplicity). In the nearest neighboring interacting case, $n_k=2$ if $k$ is an EUV.  
	Without loss of generality, let us suppose that $x^{\text{L}}_k<x^{\text{L}}_S$~(the other case $x^{\text{R}}_k>x^{\text{R}}_S$ can be treated similarly). 
	In this case we need to have $x^{\text{R}}_S\leq x^{\text{R}}_k+L-1$, because otherwise the cluster is already unembeddable before $\hat{h}_k$ is attached, contradicting the assumption that $k$ is the EUV.
	Then we have the following theorem:
	\begin{theorem}\label{thm:sum_EUV}
		\begin{equation}\label{eq:sum_EUV_NN}
		[e^{i\mathbf{H}t}(\hat{S})]_k=i\int^t_0\! dt'\, e^{i\mathbf{H}(t-t')} \mathbf{h}_{k} (e^{i\mathbf{H}'_{k}t'}\hat{S} -e^{i\mathbf{H}_{k}t'}\hat{S}),
		\end{equation}
		where 
		\begin{equation}\label{def:H_kalpha_NN}
		\hat{H}'_{k}=\hat{H}_{[x^{\text{L}}_k+1,\ldots, x^{\text{L}}_k+L]},~~~~ \hat{H}_{k}=\hat{H}_{[x^{\text{L}}_k+1,\ldots, x^{\text{L}}_k+L-1]},
		\end{equation}
		where $\hat{H}_{[x^{\text{L}}_k+1,\ldots, x^{\text{L}}_k+L]}$ denotes the truncation of $\hat{H}$ to sites $[x^{\text{L}}_k+1,\ldots, x^{\text{L}}_k+L]$, and similarly for $\hat{H}_{[x^{\text{L}}_k+1,\ldots, x^{\text{L}}_k+L-1]}$.
	\end{theorem}
	 Thm.~\eqref{thm:sum_EUV} can be proved by Taylor expanding both sides and explicitly comparing terms, using the same spirit as the proof of Lemma 5 and Lemma 6 in Ref.~\cite{Lucas2019operator}. Intuitively, $e^{i\mathbf{H}'_{k}t'}\hat{S}$ gives the sum of all terms in the Taylor expansion of $e^{i\mathbf{H}t'}\hat{S}$ that act inside the region $[x^{\text{L}}_k+1,\ldots, x^{\text{L}}_k+L]$, so $(e^{i\mathbf{H}'_{k}t'}\hat{S} -e^{i\mathbf{H}_{k}t'}\hat{S})$ gives the sum of all terms in  $e^{i\mathbf{H}t'}\hat{S}$ that act inside the region $[x^{\text{L}}_k+1,\ldots, x^{\text{L}}_k+L]$ and act non-trivially on $x^{\text{L}}_k+L$, since those terms which do not act on $x^{\text{L}}_k+L$ are canceled by $-e^{i\mathbf{H}_{k}t'}\hat{S}$. Further, only those terms in $(e^{i\mathbf{H}'_{k}t'}\hat{S} -e^{i\mathbf{H}_{k}t'}\hat{S})$ that act non-trivially on $x^{\text{R}}_k=x^{\text{L}}_k+1$ can survive the commutation with $\hat{h}_k$ in the rhs of Eq.~\eqref{eq:sum_EUV_NN}. In short, the surviving terms in $(e^{i\mathbf{H}'_{k}t'}\hat{S} -e^{i\mathbf{H}_{k}t'}\hat{S})$ are those terms that span exactly $L$ unit cells $[x^{\text{R}}_k,\ldots, x^{\text{L}}_k+L]$, which are the terms in $e^{i\mathbf{H}t'}\hat{S}$ that are $L$-site embeddable but become unembeddable immediately after attaching $\mathbf{h}_{k}$. Therefore the rhs of Eq.~\eqref{eq:sum_EUV_NN} gives the sum of all Lie clusters in $e^{i\mathbf{H}t}(\hat{S})$ with $k$ being the EUV, since the action of $e^{i\mathbf{H}(t-t')}$ happens at a time later than $t'$ and can not change the earliestness of $\hat{h}_k$. [The last sentence can be better understood by  noticing that $i\int^t_0\! dt'\, e^{i\mathbf{H}(t-t')} \mathbf{h}_{k} e^{i(\mathbf{H}-\mathbf{h}_k)t'}\hat{S}$ simply gives the sum of all terms in $e^{i\mathbf{H}t}\hat{S}$ in which $\hat{h}_k$ appears at least once, with the $\hat{h}_k$ at time $t'$ being the earliest appearance. Therefore it is natural to expect that $i\int^t_0\! dt'\, e^{i\mathbf{H}(t-t')} \mathbf{h}_{k} (e^{i\mathbf{H}'_{k}t'}\hat{S} -e^{i\mathbf{H}_{k}t'}\hat{S})$ gives a subset of the terms in $i\int^t_0\! dt'\, e^{i\mathbf{H}(t-t')} \mathbf{h}_{k} e^{i(\mathbf{H}-\mathbf{h}_k)t'}\hat{S}$, whose EUV is $k$.]

	For more general locally interacting Hamiltonians which involve terms that act non-trivially on more than 2 neighboring unit cells,  Thm.~\eqref{thm:sum_EUV} is generalized to 
	\begin{equation}\label{eq:sum_EUV}
	[e^{i\mathbf{H}t}(\hat{S})]_k=i\sum^{n_k-1}_{\alpha=1}\int^t_0\! dt'\, e^{i\mathbf{H}(t-t')} \mathbf{h}_{k} (e^{i\mathbf{H}'_{k,\alpha}t'}\hat{S} -e^{i\mathbf{H}_{k,\alpha}t'}\hat{S}),
	\end{equation}
	where 
	\begin{equation}\label{def:H_kalpha}
	\hat{H}'_{k,\alpha}=\hat{H}_{[x^{\text{L}}_k+\alpha,\ldots, x^{\text{L}}_k+\alpha+L-1]},~~~~ \hat{H}_{k,\alpha}=\hat{H}_{[x^{\text{L}}_k+\alpha,\ldots, x^{\text{L}}_k+\alpha+L-2]}.
	\end{equation}

	We can get an upper bound for the operator norm of the rhs of Eq.~\eqref{eq:sum_EUV} using  the triangle inequality and unitary invariance of operator norm. The result is 
	\begin{equation}\label{eq:EB_PBC_S1}
	\|e^{i\mathbf{H}t}\hat{S} |_{\text{span}> L}\|\leq \sum_{k,\alpha}\int^t_0\! dt'\, \|[\hat{h}_k, e^{i\mathbf{H}'_{k,\alpha}t'}(\hat{S}) -e^{i\mathbf{H}_{k,\alpha} t'}(\hat{S}) ]\|.
	\end{equation}
	The second class of terms~[those arising from $e^{i\mathbf{H}_L t}(\hat{S})$] can be bounded in a similar way, and it turns out that the resulting upper bound for the second class  is identical to the first class, Eq.~\eqref{eq:EB_PBC_S1}. Adding up the two classes and taking $n_k=2$, we get  
\begin{equation}\label{eq:EB_PBC}
|\delta \langle \hat{S}(t)\rangle|\leq 
2\sum_{k,\alpha}\int^t_0\! dt'\, \|[\hat{h}_k, e^{i\mathbf{H}'_{k,\alpha}t'}(\hat{S}) -e^{i\mathbf{H}_{k,\alpha} t'}(\hat{S}) ]\|.
\end{equation}

	In the following we give two different approaches to upper bound the commutator in the rhs of Eq.~\eqref{eq:EB_PBC_S1}. The first one is based on a combination of the methods in Ref.~\cite{Lucas2019operator,wang:tightening_2019}. This method leads to the tightest bound but has a high computational cost~(potentially exponential in the system size) except in a few simple cases. 
	The  second one is based on the differential equation method in Ref.~\cite{wang:tightening_2019}, which is slightly looser than the first one, but is much easier to compute, and can also lead to a simple analytic upper bound such as Eq.~\eqref{eq:deltaSPBC}.

	\section{Chen-Lucas bound for LR commutators}
	In this section we present the tightest LR bounds for locally interacting systems, obtained by applying the method in Ref.~\cite{Lucas2019operator} to the  commutativity graph introduced in Ref.~\cite{wang:tightening_2019}. 
	The goal is to upper bound $\|[\hat{h}_j,\hat{S}(t)]\|$ which appears in the simple bound in Eq.~\eqref{eq:deltaS0tintro} and $\|[\hat{h}_k, e^{i\mathbf{H}'_{k} t'}(\hat{S}) -e^{i\mathbf{H}_{k} t'}(\hat{S}) ]\|$ which appears in the improved PBC bound in Eq.~\eqref{eq:EB_PBC_S1}. We begin with the first one. The bound for the second one is a generalization of the first one and is qualitatively smaller at early times. For simplicity we focus on 1D in this section. 
	
	\begin{figure}
		\includegraphics[width=0.4\linewidth]{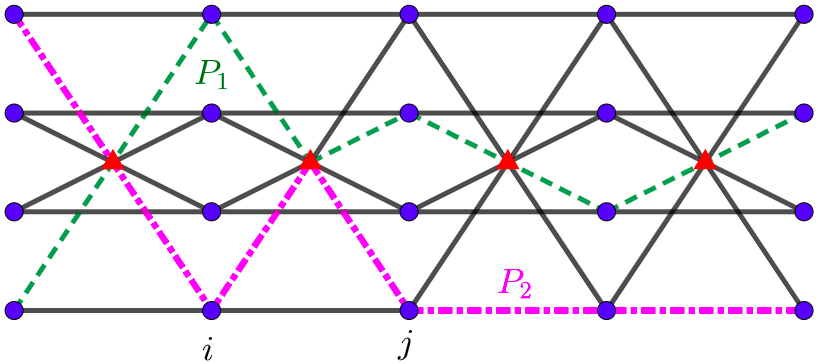}
		\includegraphics[width=0.42\linewidth]{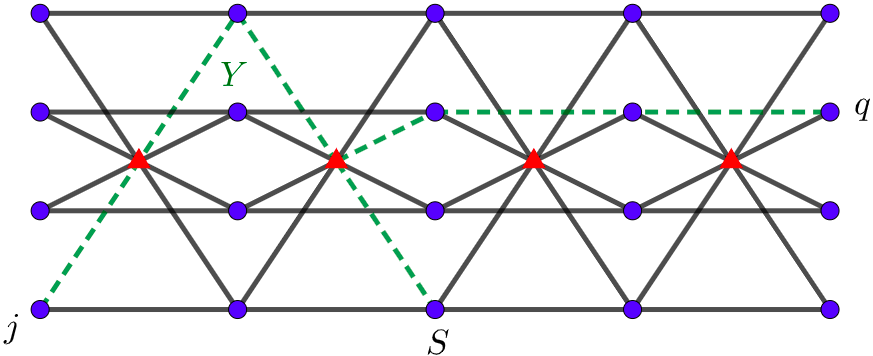}
		\caption{\label{fig:irredpath-FH} Examples of irreducible path and $Y$-shape on the commutativity graph of the FHM.  (Left) The dark green, dashed path $P_1$ is irreducible, while the pink, dot-dashed path $P_2$ is reducible since the pair $(i,j)$ is not consecutive in $P_2$ but $(i,j)\in G$. (Right) Example of an irreducible $Y$-shape.}
	\end{figure}
	The bound for $\|[\hat{h}_j,\hat{S}(t)]\|$ is obtained by generalizing Thm.~4 of Ref.~\cite{Lucas2019operator} to the commutativity graph $G$. The result is~(we present the time integrated version for simplicity)
	\begin{eqnarray}\label{eq:Lucas_bound_commu}
	\int^t_0 \|[\hat{h}_j,\hat{S}(t')]\|dt'\leq  \|\hat{S}\|\sum_{P\in\mathcal{P}_{jS}} \frac{(2t)^{n(P)-1}}{[n(P)-1]!}\prod_{i\in P,i\ne s} h_i,
	\end{eqnarray}
	where $\mathcal{P}_{jS}$ is the set of all irreducible paths on $G$ from $S$ to $j$, $n(P)$ is the number of vertices in $P$, and $h_i=\|\hat{h}_i\|$. An irreducible path $P$ on graph $G$ is a simple path~(a path without repeated vertices) in which any two non-consecutive vertices in $P$ are not  adjacent in $G$. That is, let $P=(i_{n}=j,i_{n-1},i_{n-2},\ldots,i_2,i_1=S)$, then we have $(i_m,i_{m-1})\notin G, ~2\leq m\leq n$. See Fig.~\ref{fig:irredpath-FH} for examples of irreducible paths. 
	
	The bound for $\|[\hat{h}_j, e^{i\mathbf{H}'_{j} t'}(\hat{S}) -e^{i\mathbf{H}_{j} t'}(\hat{S}) ]\|$ is obtained in a similar way:
	\begin{eqnarray}\label{eq:Lucas_bound_dbcommu}
	\int^t_0\|[\hat{h}_j, e^{i\mathbf{H}'_{j}t'}(\hat{S}) -e^{i\mathbf{H}_{j}t'}(\hat{S}) ]\| dt'\leq \|\hat{S}\|\sum_{\hat{h}_q\in \hat{H}'_{j}- \hat{H}_{j}} \sum_{Y\in\mathcal{Y}_{s,jq}}\binom{n_j(Y)+n_q(Y)-1}{n_q(Y)} \frac{(2t)^{n(Y)-1}}{[n(Y)-1]!}\prod_{i\in Y,i\ne s} h_i,
	\end{eqnarray}
	where the first sum is over all vertices $q$ that appears in $\hat{H}'_{j}- \hat{H}_{j}$, $\mathcal{Y}_{s,jq}$ is the set of all irreducible $Y$-shapes on $G$ with root $s$ and end points $j,q$,  $n_j(Y)$ is the number of vertices on the $j$-branch of $Y$ and similarly for $n_q(Y)$, so that $n(Y)=n_j(Y)+n_q(Y)+n_s(Y)+1$. The definition of an irreducible $Y$-shape with root $s$ and endpoints $j,q$ generalizes that of an irreducible path: it is a three-branch tree with a branch point $y$~($y$ may coincide with one of $s,j,q$) such that the three branches $P_{sy}, P_{jy},P_{qy}$ are irreducible paths, and 
	any vertex in $P_{sy}\backslash\{y\}$ is not linked~(with respect to $G$) to any vertex in $(P_{jy}\cup P_{qy}) \backslash\{y\}$.
	The binomial coefficient in Eq.~\eqref{eq:Lucas_bound_dbcommu} arises due to the different ways of relative time ordering the operators on the $j$-branch and the $q$-branch of $Y$.
	
	We can see that at early times, the improved PBC bound in Eq.~\eqref{eq:Lucas_bound_dbcommu} grows like $ t^{\min n(Y_{s})-1}$, where the minimum is taken over all the $L$-cell-unembeddable $Y$-shapes with root $s$ (notice that the set of smallest $L$-cell-unembeddable $Y$-shapes always contain an irreducible one). Put it another way, the small-$t$ exponent of the improved PBC bound is equal to the minimum number of Hamiltonian terms needed to be attached to $\hat{S}$ to make the cluster unembeddable~(or, the number of Hamiltonian terms in the smallest unembeddable Lie cluster containing $S$), which agrees with perturbation theory~(provided that the expectation value in $|\psi\rangle$ of the leading term doesn't vanish).
	
	In general, the rhs of Eqs.~(\ref{eq:Lucas_bound_commu},\ref{eq:Lucas_bound_dbcommu}) can only be calculated numerically. 
	Yet there are a few special cases in which we can obtain simple analytic expressions due to the simple structure of the commutativity graph. In the following we show the bound for 1D TFIM as an example. Similar bounds apply to any 1D model whose commutativity graph is a single chain, another example is the FHM in the large-$U$ limit~\cite{wang:tightening_2019}.
	\paragraph{Example: 1D TFIM in PBC}
	We write the Hamiltonian as
	\begin{equation}\label{eq:HIsing_S}
	\hat{H}=-J\sum_{j}\hat{Z}_{j,j+1}-h\sum_{j}\hat{X}_j,
	\end{equation}
	where $\hat{Z}_{j,j+1}\equiv\hat{\sigma}^z_j\hat{\sigma}^z_{j+1}$ and $\hat{X}_j\equiv \hat{\sigma}^x_j$. For illustrative purpose, take $\hat{S}=\hat{\sigma}^x_i$. The commutativity graph $G$ is simply a 1D ring, as shown in Fig.~\ref{fig:lucasTFIMPBC}. In this case, there are only two irreducible paths between any two points in $G$. Inserting Eq.~\eqref{eq:Lucas_bound_commu} evaluated for the PBC TFIM  into Eq.~\eqref{eq:deltaS0tintro} we obtain
	\begin{equation}\label{eq:TFIM_Lucas_simple}
	|\delta\langle \hat{\sigma}^x_j(t)\rangle_\psi| \leq 4\sqrt{\frac{J}{h}}\frac{(2\sqrt{Jh}t)^L}{L!}+2\sqrt{\frac{J}{h}}\frac{(2\sqrt{Jh}t)^{L+2}}{(L+2)!}, 
	\end{equation}
	where we assume for simplicity that $L$ is an odd integer.  
	
	The improved PBC error bound for $\delta\langle\hat{\sigma}^x_i(t)\rangle_\psi$ is given by Eq.~\eqref{eq:EB_PBC}, which in the current case becomes
	\begin{equation}\label{eq:EB_PBC_TFIM}
	|\delta\langle\hat{\sigma}^x_i(t)\rangle_\psi|\leq 2\sum_{|j-i|\leq L}\int^t_0\! dt'\, \|[J \hat{Z}_{j,j+1}, e^{i\mathbf{H}'_{j} t'}(\hat{X}_i) -e^{i\mathbf{H}_{j}t'}(\hat{X}_i) ]\|,
	\end{equation}
	where $\hat{H}'_{j}=\hat{H}_{[j+1,\ldots, j+L]},~ \hat{H}_{j}=\hat{H}_{[j+1,\ldots, j+L-1]}$ for $j<i$, and we have a similar expression for $j\geq i$.
	Now we apply Eq.~\eqref{eq:Lucas_bound_dbcommu} to bound the rhs. In this case, $\hat{h}_q=JZ_{j+L-1,j+L}$, and since both $\hat{H}'_{j},\hat{H}_{j}$ have open boundary, there is only one irreducible $Y$-shape with endpoints $j,s,q$. Eq.~\eqref{eq:Lucas_bound_dbcommu} becomes
	\begin{equation}\label{eq:TFIMcommutator_bound}
	\int^t_0\! dt'\,\|[J \hat{Z}_{j,j+1}, e^{i\mathbf{H}'_{j} t'}(\hat{X}_i) -e^{i\mathbf{H}_{j}t'}(\hat{X}_i) ]\|\leq \frac{(2t)^{2L-2}}{(2L-2)!}J^L h^{L-2}\binom{2L-3}{n_{j}},~~n_{j}=\begin{cases}
	2(i-j)-1, & i-L<j<i,\\
	2(j-i)+1,& i\leq j< i+L.
	\end{cases}
	\end{equation}
	When $|j-i|=L$, we have $e^{i\mathbf{H}_{j}t'}(\hat{X}_i)=\hat{X}_i$, so we can simply use Eq.~\eqref{eq:Lucas_bound_commu}. The result is similar to the first term in Eq.~\eqref{eq:TFIM_Lucas_simple} with substitution $L\to 2L-1$. Inserting Eq.~\eqref{eq:TFIMcommutator_bound} along with the $|j-i|=L$ case into Eq.~\eqref{eq:EB_PBC_TFIM}, we get
	\begin{eqnarray}\label{eq:FSETFIMPBC}
	|\delta \langle\sigma^x_j(t)\rangle|\leq 4\sqrt{\frac{J}{h}}\frac{(2\sqrt{Jh}t)^{2L-1}}{(2L-1)!}+\frac{J}{h}\frac{(4\sqrt{Jh}t)^{2L-2}}{(2L-2)!}.
	\end{eqnarray}
	
	\section{Bounding the PBC error bound by solving a linear differential equation}
	Apart from a few special cases, the computational complexity of the Chen-Lucas bound, Eq.~\eqref{eq:Lucas_bound_dbcommu}, grows exponentially with system size, since the number of irreducible paths on $G$ grows exponentially in general. 
	For some models with very complicated $G$, the computation of the Chen-Lucas bound may take even longer than the quantum dynamics simulation itself. For this reason, in this section we give an alternative method based on Ref.~\cite{wang:tightening_2019}, which is slightly looser than the Chen-Lucas bound but whose computational time complexity grows only quadratically with system size. In addition, with some further simplifications this method leads  to the simple analytic expression in Eq.~\eqref{eq:deltaSPBC}.

	The goal here is to bound the FSE of a dynamics simulation on a periodic cluster of size $L_1\times L_2\times \ldots \times L_d$. The first step is to extend the 1D bound to $d$ dimensions. To this end, denote by $\mathcal{C}_j$ a cluster of size $L_1\times L_2\times\ldots \times L_{j}\times \infty\times\ldots\times \infty$, that is, for $i\leq j$, the $i$-th direction is periodic with size $L_i$, while for $i>j$ the $i$-th direction is infinite. Then we have
	\begin{eqnarray}\label{eq:deltaS_DDim}
	|\delta \langle \hat{S}(t)\rangle_{L_1\times \ldots \times L_d}| &\equiv&|\langle \hat{S}(t)\rangle_{\infty\times \ldots \times \infty}-\langle \hat{S}(t)\rangle_{L_1\times \ldots \times L_d}|\nonumber\\
	&\leq&  \sum^d_{j=1}|\langle \hat{S}(t)\rangle_{\mathcal{C}_{j-1}}-\langle \hat{S}(t)\rangle_{\mathcal{C}_{j}}|
	\end{eqnarray}
	Since for each $j$, the clusters $\mathcal{C}_{j-1}$ and $\mathcal{C}_j$ only differ in the $j$-th direction, the difference $|\langle \hat{S}(t)\rangle_{\mathcal{C}_{j-1}}-\langle \hat{S}(t)\rangle_{\mathcal{C}_{j}}|$ can be upper bounded using the 1D method. In the following we first focus on the $j=1$ term, since other terms can be treated in an almost identical way. As before, we mainly focus on the ``nearest-neighbor interacting'' case~(only allow interactions between neighboring unit cells), as the generalization to non-nearest-neighbor interactions is straightforward. We have
	\begin{eqnarray}\label{eq:EB_PBC_S2}
	|\langle \hat{S}(t)\rangle_{\mathcal{C}_{0}}-\langle \hat{S}(t)\rangle_{\mathcal{C}_{1}}|
	\leq 2 \int^t_0\! dt'\, \sum_k\|[\hat{h}_k, e^{i\mathbf{H}'_{k} t'}(\hat{S}) -e^{i\mathbf{H}_{k} t'}(\hat{S}) ]\|.
	\end{eqnarray}
	Notice that
	\begin{eqnarray}\label{eq:dbcommu}
	[\hat{h}_k, e^{i\mathbf{H}'_{k} t}(\hat{S}) -e^{i\mathbf{H}_{k}t}(\hat{S}) ]
	&=& \int^t_0 \! dt'\, \mathbf{h}_k  e^{i\mathbf{H}'_{k}t}\frac{d}{dt'}(\hat{S} - e^{-i\mathbf{H}'_{k}t'}e^{i\mathbf{H}_{k}t'}\hat{S}) \nonumber\\
	&=&i\int^t_0 \! dt'\, \mathbf{h}_k e^{i\mathbf{H}'_{k}(t-t')} (\mathbf{H}'_{k}-\mathbf{H}_{k}) e^{i\mathbf{H}_{k} t'}(\hat{S}).
	\end{eqnarray}
	Ref.~\cite{wang:tightening_2019} introduces a method to bound unequal time commutator of the form $\mathbf{h}_{k} e^{i\mathbf{H} t}(\hat{S})$ by solving a first order linear differential equation on the commutativity graph $G$. In the following we extend this method to bound double commutators of the form $\mathbf{h}_k e^{i\mathbf{H}_2(t-t')} \mathbf{h}_{j} e^{i\mathbf{H}_1 t'}(\hat{S})$, where $\hat{H}_2=\hat{H}'_{k}, \hat{H}_1=\hat{H}_{k}$, as required for Eq.~\eqref{eq:dbcommu}.
	
	To begin, let us first recall some basic results from Ref.~\cite{wang:tightening_2019} and fix our notations. The thermodynamic limit Hamiltonian with commutativity graph $G$ is written in Eq.~\eqref{eq:H_TD} as $\hat{H}=\sum_{j\in G} \hat{h}_j$. Since both $\hat{H}_{1}$ and $\hat{H}_2$ are sums of subset of terms in  $\hat{H}_\infty=\hat{H}$~[see Eq.~\eqref{def:H_kalpha}],  we can write
	\begin{equation}
	\hat{H}_a=\sum_{j\in G} \hat{h}^a_j,~~~ a=1,2,\infty,
	\end{equation}
	where $\hat{h}^a_j=\hat{h}_j$ if the term $\hat{h}_j$ is contained in $\hat{H}_a$ and $\hat{h}^a_j=0$ otherwise. Let $G^a_{ij}(t)$ be the solution to the differential equation
	\begin{eqnarray}\label{eq:DEGreen}
	\frac{d}{dt}G^a_{ij}(t)=2\sum_{k:\langle ki\rangle\in G} \sqrt{h^a_i h^a_k} G^a_{kj}(t)\equiv \sum_{k\in G} M^a_{ik} G^a_{kj},~~a=1,2,\infty,
	\end{eqnarray}
	with initial condition $G^a_{ij}(0)=\delta_{ij}$, where $h^a_i\equiv \|\hat{h}^a_i\|\geq 0$ and $M^a_{ik}\equiv 2\sqrt{h^a_i h^a_k} \cdot (\langle ik\rangle\in G)$. The solutions can be written formally as
	$G^a_{ij}(t)=[e^{M^at}]_{ij}$. Notice that since $0\leq h^1_j\leq h^2_j\leq h_j$ we have $0\leq M^1_{ij}\leq M^2_{ij}\leq M_{ij}$ and therefore $G^{1}_{ij}(t)\leq G^{2}_{ij}(t)\leq G_{ij}(t)$ for $t\geq 0$. In translation invariant systems, $G_{ij}(t)$ has a Fourier integral expression
	\begin{equation}\label{eq:G_ijt_fourier}
	G_{ij}(t)=[e^{M t}]_{ij} =\int_{\vec{k}}~ [e^{M_{\vec{k}}t }]_{\alpha_i\alpha_j} e^{i\vec{k}\cdot(\vec{r}_i-\vec{r}_j)},
	\end{equation}
	where $\vec{r}_i$ denotes the lattice translation vector of the unit cell containing vertex $i$, $\alpha_i$ labels the index of $i$ inside a unit cell, we use the notation $\int_{\vec{k}}=\int \frac{d^dk}{(2\pi)^d}$, and
	\begin{equation}
	[M_{\vec{k}}]_{\alpha_i\alpha_j}\equiv \sum_{\vec{r}_j} M_{ij} e^{-i\vec{k}\cdot(\vec{r}_i-\vec{r}_j)}.
	\end{equation}
	In the following we first upper bound $\|\mathbf{h}_k e^{i\mathbf{H}_2(t-t')} \mathbf{h}_{j} e^{i\mathbf{H}_1 t'}(\hat{S})\|$ in terms of $G^a_{ij}(t)$ and then apply Eq.~\eqref{eq:G_ijt_fourier} to get a simple final expression. 

	First consider the case when $t'=t$. For an arbitrary local operator $\hat{A}$, denote $\hat{A}^{i}(t)=[\hat{h}_i,\hat{A}(t)]$ and $~\hat{A}^{ij}(t)=[\hat{h}_i,[\hat{h}_j,\hat{A}(t)]]$, where $\hat{A}(t)\equiv e^{i\mathbf{H}_1t} (\hat{A})$. We want to find an upper bound for $\|\hat{S}^{ij}(t)\|$. Taking the time derivative using Heisenberg's equation, we have~[note: from Eq.~\eqref{eq:HeisenbergAij} to Eq.~\eqref{eq:sol_gammaij} we write $\hat{h}_j$ to mean $\hat{h}^1_j$ for notational simplicity]
	\begin{eqnarray}\label{eq:HeisenbergAij}
	i\frac{d}{dt}\hat{A}^{ij}(t)&=&[\hat{h}_i,[\hat{h}_j,[\hat{A}(t),\sum_{k:\langle k A\rangle\in G}\hat{h}_k(t)]]]\nonumber\\
	&=&\sum_{k:\langle k A\rangle\in G}\{[\hat{A}^{ij}(t),\hat{h}_k(t)]+[\hat{A}(t),\hat{h}^{ij}_k(t)]+[\hat{A}^i(t),\hat{h}^j_k(t)]+[\hat{A}^j(t),\hat{h}^i_k(t)]\}.
	\end{eqnarray}
	We can use the same derivations as in Eqs.~(6-8) in Ref.~\cite{wang:tightening_2019} to prove that
	\begin{eqnarray}\label{eq:Aijtint}
	\|\hat{A}^{ij}(t)\|\leq 2\sum_{k:\langle k A\rangle\in G} \int^t_0 \{\|\hat{A}\|  \|\hat{h}^{ij}_k(t')\|+\|\hat{A}^i(t')\|\|\hat{h}^j_k(t')\|+\|\hat{A}^j(t')\| \|\hat{h}^i_k(t')\|\}dt'.
	\end{eqnarray}
	The last two terms can be bounded by Eqs.~(16,17) of Ref.~\cite{wang:tightening_2019}: 
	\begin{equation}\label{eq:Ait}
	\|\hat{A}^i(t)\|\leq \bar{A}^i(t)\equiv \sum_{k:\langle kA\rangle\in G} G^1_{ik}(t)\sqrt{h_k h_i}2\|\hat{A}\|.
	\end{equation}
	Notice that $\frac{d}{dt}\bar{A}^i(t)= \sum_{k:\langle kA\rangle\in G}2\|A\|\bar{h}_k^i(t)$, so the last two terms in Eq.~\eqref{eq:Aijtint} can be combined to $\frac{d}{dt'}[\bar{A}^i(t')\bar{A}^i(t')]/\|\hat{A}\|$, and we get
	\begin{equation}\label{eq:Aijt}
	\|\hat{A}^{ij}(t)\|\leq 2\sum_{k:\langle k A\rangle\in G} \int^t_0 \|\hat{A}\|  \|\hat{h}^{ij}_k(t')\|dt' + \bar{A}^i(t)\bar{A}^j(t)/\|\hat{A}\|.
	\end{equation}
	Now take $\hat{A}=\hat{h}_l$, a term in  $\hat{H}$. Using Gr\"{o}nwall's inequality, we can prove that  $\|\hat{h}^{ij}_l(t)\|\leq \bar{h}^{ij}_l(t)$ where $\bar{h}^{ij}_l(t)$ is the solution to the differential equation
	\begin{equation}\label{eq:DEgamma_ijl}
	\frac{d}{dt}\bar{h}^{ij}_l(t)= 2\sum_{k:\langle k l\rangle\in G} h_l  \bar{h}^{ij}_k(t) + \frac{1}{h_l}\frac{d}{dt}[\bar{h}^i_l(t)\bar{h}^j_l(t)],
	\end{equation}
	with initial condition $\bar{h}^{ij}_l(0)=0$. If we substitute $\bar{\Gamma}_l(t)=\bar{h}^{ij}_l(t)/\sqrt{h_l}$, Eq.~\eqref{eq:DEgamma_ijl} is of the form $\frac{d}{dt}\bar{\Gamma}=M^a\cdot \bar{\Gamma}+B$, which has formal solution $\bar{\Gamma}(t)=\int_0^t \! dt'\,e^{M^a(t-t')}B(t')dt'$, i.e.
	\begin{equation}\label{eq:sol_gammaij}
	\bar{h}^{ij}_l(t)=\sum_k \int^t_0 \sqrt{\frac{h_l}{h^3_k}}G^1_{lk}(t-t')d[\bar{h}^i_k(t')\bar{h}^j_k(t')].
	\end{equation}
	An upper bound for $\| S^{ij}(t)\|$ can be obtained by taking $\hat{A}=\hat{S}$ in Eq.~\eqref{eq:Aijt} and inserting Eqs.~(\ref{eq:Ait}, \ref{eq:sol_gammaij}). 
	Finally, Eq.~(15) in Ref.~\cite{wang:tightening_2019} allows us to upper bound $\|\boldsymbol{h}_i  e^{i\mathbf{H}_2(t-t')} \boldsymbol{h}_{j} e^{i\mathbf{H}_1 t'}(\hat{S})\|$ in terms of $\|\hat{S}^{ij}(t)\|$ and $G^2_{ij}(t)$ by taking $\hat{B}=\boldsymbol{h}_{j} e^{i\mathbf{H} t'}(\hat{S})$:
	\begin{equation}\label{eq:dbcommufinal}
	\|\boldsymbol{h}_i  e^{i\mathbf{H}_2(t-t')} \boldsymbol{h}_{j} e^{i\mathbf{H}_1 t'}(\hat{S})\|\leq 2 \sum_{l\in G}G^2_{il}(t-t')\sqrt{\frac{h_i}{h_l}}\|\hat{S}^{lj}(t')\|.
	\end{equation}
	
	In summary, to get a bound for FSE $|\delta \langle \hat{S}(t)\rangle_{L_1\times \ldots \times L_d}|$, one need to first solve the differential equation Eq.~\eqref{eq:DEGreen} to get $G^1_{ij}(t),G^2_{ij}(t)$~[note that since $G^1_{ij}(t)\leq G^2_{ij}(t)$, having a bound for $G^2_{ij}(t)$ is enough], then use Eqs.~(\ref{eq:Ait},\ref{eq:Aijt},\ref{eq:sol_gammaij}) to get a bound for $\|\hat{S}^{lj}(t')\|$, then insert into Eq.~\eqref{eq:dbcommufinal} to get a bound for the double commutator, and finally use Eqs.~(\ref{eq:deltaS_DDim},\ref{eq:EB_PBC_S2},\ref{eq:dbcommu}) to bound $|\delta \langle \hat{S}(t)\rangle_{L_1\times \ldots \times L_d}|$. All steps in this procedure are efficient, with total computational cost scaling at most quadratically with the system size.
	\subsection{Derivation of Eq.~\eqref{eq:deltaSPBC}}\label{sec:derive_simple_analytic}
	We now derive the simple bound in Eq.~\eqref{eq:deltaSPBC}.
	We use $G^a_{ij}(t)\leq G_{ij}(t)$ and Eq.~\eqref{eq:G_ijt_fourier} to simplify the expression. Eq.~\eqref{eq:sol_gammaij} becomes
	\begin{eqnarray}\label{eq:sol_gammaij_Fourier}
	\bar{h}^{ij}_l(t)=\sqrt{h_jh_jh_l}\sum_{m\in G} \int_0^t \frac{1}{\sqrt{h_m}}[e^{M(t-t')}]_{lm} [M^2 e^{Mt'}]_{im}[M e^{Mt'}]_{jm} dt' +(i\leftrightarrow j). 
	\end{eqnarray}
	Taking $\hat{A}=\hat{S}$ in Eq.~\eqref{eq:Aijt} and inserting Eqs.~(\ref{eq:Ait}, \ref{eq:sol_gammaij_Fourier}), we get
	\begin{eqnarray}\label{eq:Sij_Fourier}
	\bar{S}^{ij}(t_3)=\sqrt{h_ih_j}\sum_{\substack{m,l\in G}} \int_{t_{1,2}} \frac{S_l}{\sqrt{h_m}}[e^{M(t_2-t_1)}]_{lm} [M^2 e^{Mt_1}]_{im}[M e^{Mt_1}]_{jm}  +(i\leftrightarrow j)+\bar{S}^i(t_3)\bar{S}^j(t_3)/S,
	\end{eqnarray}
	where $\int_{t_{1,2}}\equiv \int_{0\leq t_1\leq t_2\leq t_3}dt_1dt_2 $ and $S_l\equiv 2\sqrt{h_l}(\langle ls\rangle\in G)\|\hat{S}\|$. Notice that $M$ is symmetric, and so are $e^{Mt},Me^{Mt}$, etc.
	Inserting Eq.~\eqref{eq:Sij_Fourier} into Eq.~\eqref{eq:dbcommufinal}, we obtain
	\begin{eqnarray}\label{eq:dbcommufinal_Fourier}
	\|\boldsymbol{h}_i  e^{i\mathbf{H}_2(t-t_3)} \boldsymbol{h}_{j} e^{i\mathbf{H}_1 t_3}(\hat{S})\|
	&\leq& 2\int_{t_{1,2}}\sum_{\substack{m,l\in G\\a=1,2}}\sqrt{\frac{h_i h_j} {h_m}}[M^{3-a} e^{M(t-t_3+t_2-t_1)}]_{im} [M^a e^{M(t_2-t_1)}]_{jm} [e^{M t_1}]_{ml} S_l+2\bar{S}^i(t)\bar{S}^j(t_3)/S\nonumber\\
	&=& 2\int_{t_{1,2},\vec{k}_{1,2}} \sum_{\substack{\alpha_m,l\\a=1,2}}\sqrt{\frac{h_i h_j} {h_m}}[M_{\vec{k}_1}^{3-a} e^{M_{\vec{k}_1}(t-t_3+t_2-t_1)}]_{\alpha_i\alpha_m} [M^a_{\vec{k}_2} e^{M_{\vec{k}_2}(t_2-t_1)}]_{\alpha_j\alpha_m} \nonumber\\
	&&\times [e^{M_{\vec{k}_1+\vec{k}_2} t_1}]_{\alpha_m\alpha_l}  e^{i\vec{k}_1\cdot (\vec{r}_i-\vec{r}_l)+i\vec{k}_2\cdot (\vec{r}_j-\vec{r}_l)} S_l+2\bar{S}^i(t)\bar{S}^j(t_3)/S,
	\end{eqnarray}
	where $\alpha_m$ runs over all vertices in a unit cell. Now we insert Eq.~\eqref{eq:dbcommufinal_Fourier} into Eqs.~\eqref{eq:dbcommu} and \eqref{eq:EB_PBC_S2} to bound FSE $|\langle \hat{S}(t)\rangle_{\mathcal{C}_{0}}-\langle \hat{S}(t)\rangle_{\mathcal{C}_{1}}|$. Eqs.~\eqref{eq:dbcommu} and \eqref{eq:EB_PBC_S2} combines to be
	\begin{equation}\label{eq:deltaS_dbcommu}
	|\langle \hat{S}(t)\rangle_{\mathcal{C}_{0}}-\langle \hat{S}(t)\rangle_{\mathcal{C}_{1}}|
	\leq 2\sum_{(i,j)\in \mathcal{S}}\int^t_0\! dt_4\, \int^{t_4}_0 \|\boldsymbol{h}_i  e^{i\mathbf{H}_2(t_4-t_3)} \boldsymbol{h}_{j} e^{i\mathbf{H}_1 t_3}(\hat{S})\|dt_3, 
	\end{equation}
	where 
	\begin{eqnarray}\label{def:mathcalS}
	\mathcal{S}&\equiv& \{(i,j)|n_i=n_j=2~\land~ x^{\text{L}}_i<x^{\text{L}}_S~\land~ x^{\text{R}}_S\leq x^{\text{R}}_i+L-1~\land~x^{\text{L}}_j=x^{\text{L}}_i+L-1\}  \nonumber\\
	&&\cup\{(i,j)|n_i=n_j=2~\land~ x^{\text{R}}_i>x^{\text{R}}_S~\land~ x^{\text{L}}_S\geq x^{\text{L}}_i-(L-1)~\land~ x^{\text{R}}_j=x^{\text{R}}_i-(L-1)\}.
	\end{eqnarray}
	The restriction on the summation over $(i,j)$ follows from the discussion above Thm.~\ref{thm:sum_EUV} and the definition of $\hat{H}_i, \hat{H}'_i$ in Eq.~\eqref{def:H_kalpha_NN}, which is essentially the requirement that  $i$ is the EUV of some Lie cluster starting from the vertex $s$, and that $\hat{h}_j$ is a term in $\hat{H}_2-\hat{H}_1\equiv \hat{H}'_{i}-\hat{H}_{i}$. Notice that the set $\mathcal{S}$ is translation invariant in directions $L_2,\ldots,L_d$. Namely, if $(i,j)\equiv((\vec{r}_i,\alpha_i),(\vec{r}_j,\alpha_j))\in \mathcal{S}$, then for any $\vec{r}_{1\perp}\perp \hat{x},\vec{r}_{2\perp}\perp \hat{x}$, we have $(i',j')\equiv((\vec{r}_i+\vec{r}_{1\perp},\alpha_i),(\vec{r}_j+\vec{r}_{2\perp},\alpha_j))\in \mathcal{S}$. Let $\mathcal{T}_\perp$ denote the group of all lattice translation vectors in directions $L_2,\ldots,L_d$, such that $\mathcal{S}\cong \mathcal{S}/\mathcal{T}_\perp\times \mathcal{T}_\perp$. We can therefore decompose the sum over $i,j$ as $\sum_{(i,j)\in \mathcal{S}}=\sum_{(x_i,x_j)\in \mathcal{S}/\mathcal{T}_\perp}\sum_{(\vec{r}_{i\perp},\vec{r}_{j\perp})\in \mathcal{T}_\perp}$. We have~(we abbreviate $x_i=x^{\text{L}}_i,x_j=x^{\text{L}}_j$)
	\begin{eqnarray}
	\sum_{\substack{(\vec{r}_{i\perp},\vec{r}_{j\perp})\\ \in \mathcal{T}_\perp}}\|\boldsymbol{h}_i  e^{i\mathbf{H}_2(t-t_3)} \boldsymbol{h}_{j} e^{i\mathbf{H}_1 t_3}(\hat{S})\|
	&\leq& 2\sqrt{h_{\alpha_i}h_{\alpha_j}}\int_{t_{1,2},k_{1x},k_{2x}}\sum_{\substack{\alpha_m,l\\a=1,2}}[M_{k_{1x}}^{3-a} e^{M_{k_{1x}}(t-t_3+t_2-t_1)}]_{\alpha_i\alpha_m} [M^a_{k_{2x}} e^{M_{k_{2x}}(t_2-t_1)}]_{\alpha_j \alpha_m} \nonumber\\
	&&\times [e^{M_{k_{1x}+k_{2x}} t_1}]_{\alpha_m \alpha_l} \frac{1}{\sqrt{h_{\alpha_m}}}  e^{i k_{1x}(x_i-x_l)+i k_{2x} (x_j-x_l)} S_l+2\bar{S}_x^i(t)\bar{S}_x^j(t_3)/S,
	\end{eqnarray}
	where $M_{k_x}\equiv M_{\vec{k}=(k_x,0,\ldots,0)} $, and
	\begin{equation}
	\bar{S}_x^j(t)=\sum_{\vec{r}_{j\perp}\in \mathcal{T}_\perp} \bar{S}^j(t)=\sum_{\langle nS\rangle}\int_{k_x}[e^{M_{k_x}t}]_{\alpha_j\alpha_n} e^{i k_x(x_j-x_n)}2\sqrt{h_jh_n}S.
	\end{equation}
	We now need to sum over $x_i, x_j$ satisfying restrictions defined in Eq.~\eqref{def:mathcalS}. Notice that all such $(x_i,x_j)$ satisfy $|x_i-x_j|=L-1$. It turns out to be more convenient to extend the summation to all $(x_i,x_j)$ satisfying $|x_i-x_j|=L-1$. This still gives an upper bound since the rhs of Eq.~\eqref{eq:dbcommufinal_Fourier} is always non-negative. We let $x_j=x_i\pm(L-1)$ and sum over $x_i$ from $-\infty$ to $\infty$:
	\begin{eqnarray}\label{eq:direction_1_k}
	\sum_{(i,j)\in\mathcal{S}}\|\boldsymbol{h}_i  e^{i\mathbf{H}_2(t_4-t_3)} \boldsymbol{h}_{j} e^{i\mathbf{H}_1 t_3}(\hat{S})\| 
	&\leq &2\sum_{\substack{\alpha_m,\alpha_i,\alpha_j,l\\a=1,2\\\Delta x=\pm (L-1)}}\sqrt{h_{\alpha_i}h_{\alpha_j}}\int_{t_{1,2},k_{x}}[M_{-k_{x}}^{3-a} e^{M_{-k_{x}}(t_4-t_3+t_2-t_1)}]_{\alpha_i\alpha_m} [M^a_{k_{x}} e^{M_{k_{x}}(t_2-t_1)}]_{\alpha_j\alpha_m} \nonumber\\
	&&\times [e^{M_{0} t_1}]_{\alpha_m\alpha_l} \frac{1}{\sqrt{h_{\alpha_m}}} e^{i k_{x}\Delta x} S_l +\sum_{(i,j)\in\mathcal{S}/\mathcal{T}_\perp}2\bar{S}_x^i(t_4)\bar{S}_x^j(t_3)/S.
	\end{eqnarray}
	This is the bound for the $j=1$ term in Eq.~\eqref{eq:deltaS_DDim}. Those $j>1$ terms in Eq.~\eqref{eq:deltaS_DDim} can be treated in an almost identical way; the only difference is that since the directions $L_1,\ldots,L_{j-1}$ are now periodic, the integrals over $\vec{k}$ in those $j-1$ directions have to be replaced by a discrete sum, $k_m\in \{\frac{2\pi n}{L_m}| n=0,1,\ldots, L_m-1\}, 1\leq m\leq j-1$. The result, however, remains completely the same as Eq.~\eqref{eq:direction_1_k}. In summary, we have
	\begin{eqnarray}\label{eq:direction_all_k}
	|\delta \langle \hat{S}(t)\rangle_{L_1\times \ldots \times L_d}|
	&\leq &4\sum_{\substack{\alpha_m,\alpha_i,\alpha_j,l\\ 1\leq p\leq d,a=1,2\\\Delta x_p=\pm (L_p-1)}}\sqrt{h_{\alpha_i}h_{\alpha_j}}\int_{t_{1,2,3,4},k_{p}}[M_{-k_{p}}^{3-a} e^{M_{-k_{p}}(t_4-t_3+t_2-t_1)}]_{\alpha_i\alpha_m} [M^a_{k_{p}} e^{M_{k_{p}}(t_2-t_1)}]_{\alpha_j\alpha_m} \nonumber\\
	&&\times [e^{M_{0} t_1}]_{\alpha_m\alpha_l} \frac{1}{\sqrt{h_{\alpha_m}}} e^{i k_{p}\Delta x_p} S_l+4\sum_{\substack{(i,j)\in\mathcal{S}_p\\1\leq p\leq d}}\int_{t_{3,4}}\bar{S}_p^i(t_4)\bar{S}_p^j(t_3)/S,
	\end{eqnarray}
	where the time integrations are restricted to $0<t_1<t_2<t_3<t_4<t$.
	We now use the same derivation in Eq.~(28) of Ref.~\cite{wang:tightening_2019} to upper bound the $k$ integral by its analytic continuation to $i\kappa$. We focus on the first term, since the second term can be treated in a similar way.  We have
	\begin{eqnarray}\label{eq:direction_all_kappa}
	|\delta \langle \hat{S}(t)\rangle_{L_1\times \ldots \times L_d}|
	&\leq &4\sum_{\substack{\alpha_m,\alpha_i,\alpha_j,l\\ 1\leq p\leq d,a=1,2}}\sqrt{h_{\alpha_i}h_{\alpha_j}}\int_{t_{1,2,3,4}}\{[M_{i\kappa_{p}}^{3-a} e^{M_{i\kappa_{p}}(t_4-t_3+t_2-t_1)}]_{\alpha_m \alpha_i} [M^a_{i\kappa_{p}} e^{M_{i\kappa_{p}}(t_2-t_1)}]_{\alpha_j\alpha_m}+(\kappa_p\to-\kappa_p)\} \nonumber\\
	&&\times [e^{M_{0} t_1}]_{\alpha_m\alpha_l} \frac{1}{\sqrt{h_{\alpha_m}}} e^{- \kappa_{p}(L_p-1)} S_l+(\text{Term 2})\\
	&\leq &4\sum_{\substack{\alpha_i,\alpha_j,l\\ 1\leq p\leq d,a=1,2}}\sqrt{h_{\alpha_i}h_{\alpha_j}}\int_{t_{1,2,3,4}}\{[M^a_{i\kappa_{p}} e^{M_{i\kappa_{p}}(t_2-t_1)} D M_{i\kappa_{p}}^{3-a} e^{M_{i\kappa_{p}}(t_4-t_3+t_2-t_1)}]_{\alpha_j \alpha_i} +(\kappa_p\to-\kappa_p)\} \nonumber\\
	&&\times \|e^{M_{0} t_1}\|  e^{- \kappa_{p}(L_p-1)} S_l+(\text{Term 2})\nonumber\\
	&\leq &8\sum_{\substack{ 1\leq p\leq d}}c_H c_S  c_{\kappa_p}\int_{t_{1,2,3,4}} \omega_{i\kappa_p}^3 e^{\omega_{i\kappa_p}(t_4-t_3+2t_2-2t_1)+\omega_0 t_1-\kappa_p(L_p-1)} +(\text{Term 2})\nonumber\\
	&\leq &2\sum_{\substack{ 1\leq p\leq d}}c_H c_S  c_{\kappa_p} \frac{e^{2\omega_{i\kappa_p}t-\kappa_p(L-1)}}{2\omega_{i\kappa_p}-\omega_0}+(\text{Term 2}),
	\end{eqnarray}
	where $D=\mathrm{diag}\{\frac{1}{\sqrt{h_{\alpha_m}}}\}$, $c_H=\sum_{\alpha_i} h_{\alpha_i}$, $c_S=\sum_{\langle lS\rangle} 2\sqrt{h_l}S$, $c_{\kappa_p}=\|U_{i\kappa_p}\|\|U^{-1}_{i\kappa_p}\|\|U^{-1}_{i\kappa_p}DU_{i\kappa_p}\|+(\kappa_p\to-\kappa_p)$, $\omega_{i\kappa_{p}}$ is the Perron-Frobenius eigenvalue of $M_{i\kappa_p}$~(i.e. the eigenvalue with the largest magnitude; this must be real and positive by the Perron-Frobenius theorem), and $U_{i\kappa_p}$ is the matrix that diagonalizes $M_{i\kappa_p}$, i.e. $M_{i\kappa_p}=U_{i\kappa_p}\Omega_{i\kappa_p}U^{-1}_{i\kappa_p}$ for some diagonal matrix $\Omega_{i\kappa_p}$. In the third line we use Cauchy-Schwarz inequality $\sum_{i,j} \sqrt{h_i}C_{ij}\sqrt{h_j}\leq \sum_i h_i \|C\|$. Combined with the second term, the final result is
	\begin{equation}\label{eq:deltaSikappa}
	|\delta \langle \hat{S}(t)\rangle_{L_1\times \ldots \times L_d}|\leq \sum_{\substack{ 1\leq p\leq d}}C_{\kappa_p} e^{2\omega_{i\kappa_p}t-\kappa_p(L-1)},
	\end{equation}
	where 
	\begin{equation}\label{eq:const_Cp}
	C_{\kappa_p}=2 \frac{c_H c_S  c_{\kappa_p}}{2\omega_{i\kappa_p}-\omega_0}+8 Sc_H \frac{\|U_{i\kappa_p}\|^2\|U^{-1}_{i\kappa_p}\|^2}{\omega_{i\kappa_p}^2}s_{\kappa_p}s_{-\kappa_p}, ~~s_{\kappa_p}=\sqrt{\sum_{\langle lS\rangle}h_l e^{2\vec{\kappa}_p \cdot (\vec{r}_l-\vec{r}_s)}}. 
	\end{equation}
	Notice that $\omega_{i\kappa_p}$ means $\omega_{i\vec{\kappa}_p}$ where the only nonzero component of $\vec{\kappa}_p$ is  in the $p$-th direction equal to $\kappa_p$.
	
	With Eq.~\eqref{eq:deltaSikappa} we are ready to derive the bound in Eq.~\eqref{eq:deltaSPBC}. The arguments in Sec.~IV of Ref.~\cite{wang:tightening_2019} can be generalized to prove the following 
	\begin{proposition}\label{prop:1}
		Let $f(t)$ be a function with Taylor expansion $f(t)=\sum_{n\geq m} f_n t^n$, where $m$ is a positive integer and $f_n\geq 0, \forall n\geq m$. If $f(l/v)\leq C$, then $f(t)\leq C(vt/l)^m$, for $ 0\leq t\leq l/v$.
	\end{proposition}
	Now take $f(t)=|\langle \hat{S}(t)\rangle_{\mathcal{C}_{p-1}}-\langle \hat{S}(t)\rangle_{\mathcal{C}_{p}}|$. Eq.~\eqref{eq:deltaSikappa} proves that $f(\frac{L_p}{2v_p})\leq C_{\kappa_{p0}}e^{\kappa_{p0}}$~(one has to redo the above derivations for each direction separately), where 
	\begin{equation}\label{eq:directionalv_p}
	v_{p}=\min_{\kappa_p>0}\frac{\omega_{i\kappa_p}}{\kappa_p}
	\end{equation}
	is the LR speed in the $p$-th direction, and $\kappa_{p0}$ is the position of the minimum. On the other hand, using the Taylor expansion of $G_{ij}(t)$ in Eq.~\eqref{eq:G_ijt_fourier}, one can show that the upper bound for $f(t)$ given in Eqs.~(\ref{eq:dbcommufinal}) and (\ref{eq:deltaS_dbcommu}) has a Taylor series expansion  with the same leading term as the Chen-Lucas bound in Eq.~\eqref{eq:Lucas_bound_dbcommu} and with non-negative coefficients, i.e. $f(t)=\sum_{n\geq m} f_n t^n$ where $f_n\geq 0$, and $m={\min n_p(Y_{S})-1}$ is the number of Hamiltonian terms of the smallest  $Y$-shape starting from $S$ that is $L_p$-cell unembeddable in the $p$-th direction. Therefore, Proposition~\ref{prop:1} says 
	\begin{equation}
	|\langle \hat{S}(t)\rangle_{\mathcal{C}_{0}}-\langle \hat{S}(t)\rangle_{\mathcal{C}_{1}}|\leq C_{\kappa_{p0}}e^{\kappa_{p0}}\left(\frac{2v_p t}{L_p}\right)^{\min n_p(Y_{S})-1}.
	\end{equation}
	In summary, we have 
	\begin{equation}\label{eq:deltaS_simple_final}
	|\delta \langle \hat{S}(t)\rangle_{L_1\times \ldots \times L_d}|\leq \sum_{\substack{ 1\leq p\leq d}} C_p\left(\frac{2v_pt}{L_p}\right)^{\mathcal{L}_p},
	\end{equation}
	where $C_p\equiv C_{\kappa_{p0}}e^{\kappa_{p0}}$ and $\mathcal{L}_p\equiv \min n_p(Y_{S})-1$.
	In general, the exponent $\mathcal{L}_p$ is linearly related to $L_p$, i.e. $\mathcal{L}_p=\eta_p L_p-\mu_{p,S}$. This finishes the proof of Eq.~\eqref{eq:deltaSPBC}.
	
	\bibliographyA{finite-size-error}
	\bibliographystyleA{apsrev4-1}
\end{document}